\newif\ifabstract
\abstracttrue
 \abstractfalse 
\newif\iffull
\ifabstract \fullfalse \else \fulltrue \fi

\documentclass[11pt,onecolumn]{article}
\usepackage{amsfonts}
\usepackage{amssymb}
\usepackage{amstext}
\usepackage{amsmath}
\usepackage{xspace}
\usepackage{theorem}
\usepackage{graphicx}
\usepackage{url}
\usepackage{graphics}
\usepackage{colordvi}
\usepackage{colordvi}
\usepackage{subfigure}

\usepackage{setspace}
\usepackage{color}
\usepackage{cite}
\usepackage[disable]{todonotes}
\usepackage{hhline}
\usepackage{array}
\usepackage{pifont}
\usepackage{enumerate}
\usepackage{bm}

\usepackage{graphicx,epsfig,amsfonts,bbm,epstopdf,tabularx}

\usepackage{url}

\usepackage{xcolor}
\usepackage[comma,numbers,square,sort&compress]{natbib}
\usepackage{algorithm,algpseudocode}
\usepackage{algorithmicx}
\usepackage{stmaryrd}
\usepackage{multirow}
\usepackage{graphicx}
\usepackage{arydshln}
\usepackage{mathtools}
\usepackage{tabulary}
\usepackage{booktabs}
\usepackage{subfigure}
\usepackage{xspace}

\advance \oddsidemargin by -0.8in
\newcommand{\myparskip}{3pt}
\parskip \myparskip

\algnewcommand{\LeftComment}[1]{\Statex \(\triangleright\) #1}

\newcolumntype{L}[1]{>{\raggedright\let\newline\\\arraybackslash\hspace{0pt}}m{#1}}
\newcolumntype{C}[1]{>{\centering\let\newline\\\arraybackslash\hspace{0pt}}m{#1}}
\newcolumntype{R}[1]{>{\raggedleft\let\newline\\\arraybackslash\hspace{0pt}}m{#1}}

        {\hspace*{\fill}$\Box$\par\vspace{4mm}}

\ifodd 0
\newcommand{\com}[1]{\textbf{\color{blue} (COMMENT: #1)}}
\else

\newcommand{\com}[1]{}
\fi

\newcommand{\rsgenco}{\textsc{srGENCO}\xspace}
\newcommand{\scsp}{\textsf{sOSP}\xspace}
\newcommand{\socs}{\textsf{sOffer}\xspace}
\newcommand{\cratio}{\textsc{cr}\xspace}
\newcommand{\ofa}{\textsf{OFA}\xspace}
\newcommand{\ocsmb}{\textsf{mOffer}\xspace}
\newcommand{\mocsmb}{\textsf{gOffer}\xspace}
\newcommand{\nostorage}{\textsf{NoStorage}\xspace}
\newcommand{\fonline}{\textsf{FixedOnline}\xspace}

\newtheorem{myDef}{Definition}

\newcommand{\csp}{{\sf HA-CSP}\xspace}

\newcommand{\be}{\begin{enumerate}}
\newcommand{\ee}{\end{enumerate}}
\newcommand{\bd}{\begin{description}}
\newcommand{\ed}{\end{description}}
\newcommand{\bi}{\begin{itemize}}
\newcommand{\ei}{\end{itemize}}

\newtheorem{theorem}{Theorem}[section]
\newtheorem{lemma}[theorem]{Lemma}

\newenvironment{proof}{\par \smallskip{\bf Proof:}}{\hfill\stopproof}
\def\stopproof{\square}
\def\square{\vbox{\hrule height.2pt\hbox{\vrule width.2pt height5pt \kern5pt
\vrule width.2pt} \hrule height.2pt}}

\renewcommand{\phi}{\varphi}

\mathchardef\hyphen="2D

\begin{document}

\title{Online Offering Strategies for Storage-Assisted Renewable Power Producer in Hour-Ahead Market}
\author{Lin Yang\thanks{Information Engineering Department, The Chinese University of Hong Kong. Email: {\tt yl015@ie.cuhk.edu.hk}.} \and Mohammad Hassan Hajiesmaili \thanks{Institute of Network Coding, The Chinese University of Hong Kong. Email: {\tt hajiesmaili@gmail.com}.} \and Hanling Yi \thanks{Information Engineering Department, The Chinese University of Hong Kong. Email: {\tt yh014@ie.cuhk.edu.hk}.} \and Minghua Chen \thanks{Information Engineering Department, The Chinese University of Hong Kong. Email: {\tt minghua@ie.cuhk.edu.hk}.}}

\maketitle

\begin{abstract}
A promising approach to hedge against the inherent uncertainty of renewable generation is to equip the renewable plants with energy storage systems.
This paper focuses on designing profit maximization offering strategies, \textit{i.e.}, the strategies that determine the offering price and volume, for a storage-assisted renewable power producer that participates in hour-ahead electricity market. Designing the strategies is challenging since (i) the underlying problem is coupled across time due to the evolution of the storage level, and (ii) inputs to the problem including the renewable output and market clearing price are unknown when submitting offers. Following the competitive online algorithm design approach, we first study a basic setting where the renewable output and the clearing price are known for the next hour. We propose \socs, a simple online offering strategy that achieves the best possible competitive ratio of $O(\log \theta)$, where $\theta$ is the ratio between the maximum and the minimum clearing prices.
Then, we consider the case where the clearing price is unknown. By exploiting the idea of submitting multiple offers to combat price uncertainty, we propose \ocsmb, and demonstrate that the competitive ratio of \ocsmb converges to that of \socs as the number of offers grows.
Finally, we extend our approach to the scenario where the renewable output has forecasting error. We propose \mocsmb as the generalized offering strategy and characterize its competitive ratio as a function of the forecasting error. Our trace-driven experiments demonstrate that our algorithms achieve performance close to the offline optimal and outperform a baseline alternative significantly.
\end{abstract}

\section{Introduction}
Nowadays, renewable power producers, such as wind farms and solar plants, are being rapidly integrated to the power system. In $2015$, investment on renewables sets a record of $296$ billion dollars, more than double the amount for fossil fuels~\cite{ren21}.
Renewables are attractive in that they are clean, free (except capital and maintenance cost), and inexhaustible. Integration of renewables into the power system and particularly in the electricity market, however, is challenging since their generation is uncontrollable, intermittent, and unpredictable.
A promising approach to facilitate the renewable integration and hedge against the uncertainty, is to equip the renewable plants with the giant energy storage systems~\cite{dunn2011electrical,Chowdhury2016Benefits}. Some examples are the storage stations at Southern California (with capacity of $40$MWh), South Korea ($16$MWh), and Germany ($15$MWh)~\cite{ex_battery}.


In several existing electricity markets, renewable power producers receive guaranteed feed-in tariffs, \textit{e.g.}, ``take-all-wind'' policy in California's electricity market~\cite{caiso}. ``take-all-wind'' policy guarantees that the market absorbs the entire renewable supply at favorable fixed prices. This policy is feasible since the current market share of the renewables is not significant, at most $19.2\%$, on average around the world, in $2015$~\cite{ren21}.

The extra-market treatment policy, however, is not viable in future. Recently, the deployment cost of the renewables is dropping rapidly, \textit{e.g.}, the price of solar panel has fallen by $80\%$ from $2008$ to $2016$~\cite{bloomberg}. This reduction in price along with the environmental concerns
pushes rapid penetration of the renewables, \textit{e.g.}, in Denmark, the plan is to achieve $50\%$ and $100\%$ renewable generation in $2020$ and $2050$, respectively \cite{Denmarkwind2}. Hence, in eventual market with the considerable renewables' market share, it is inevitable to treat renewable producers the same as other traditional generation companies~\cite{kim2011optimal}.

Several electricity markets such as  NYISO~\cite{NYISO}, CAISO~\cite{caiso}, and Nord Pool~\cite{nord} operate in a multi-settlement manner and settle transactions at multiple timescales, \textit{i.e.}, day-ahead, hour-ahead~\cite{Harvard2011Primer,jiang2015optimal}, and real-time.
Considering the uncertainty of renewable output, \rsgenco tends to participate in short-term market, specifically hour-ahead market, without suffering profit reduction caused by long-term forecasting errors \cite{kim2011optimal}. In reality, CAISO's Participating Intermittent Resource Program (PIRP \cite{caiso}) is already requiring wind power plants to bid into their hour-ahead market. In hour-ahead market operation, the generation companies including renewable producers submit their \textit{offers} (including offering price and offering volume) for selling the electricity in the next hour (see Sec.~\ref{sec:hour} for more details on the hour-ahead market operation). This work focuses on designing offering strategies, \textit{i.e.}, the strategies that determine the offering price and volume, for Storage-assisted Renewable GENeration COmpany (\rsgenco) that participates in hour-ahead market.

As depicted in Fig.~\ref{fig:scenario}, we consider a scenario in which an \rsgenco, like other traditional generation companies, participates in hour-ahead market by (1) submitting the offer. After receiving the offers, (2) the market operator matches the offers with the bids from the demand-side and announces a market clearing price. If the offering price of \rsgenco is less than the clearing price, (3) its offering volume is considered as the \textit{commitment} to the market for the next hour. In turn, \rsgenco is paid according to the clearing price.
%

\begin{figure}[h]
	\begin{center}
		\includegraphics[width=0.8\textwidth]{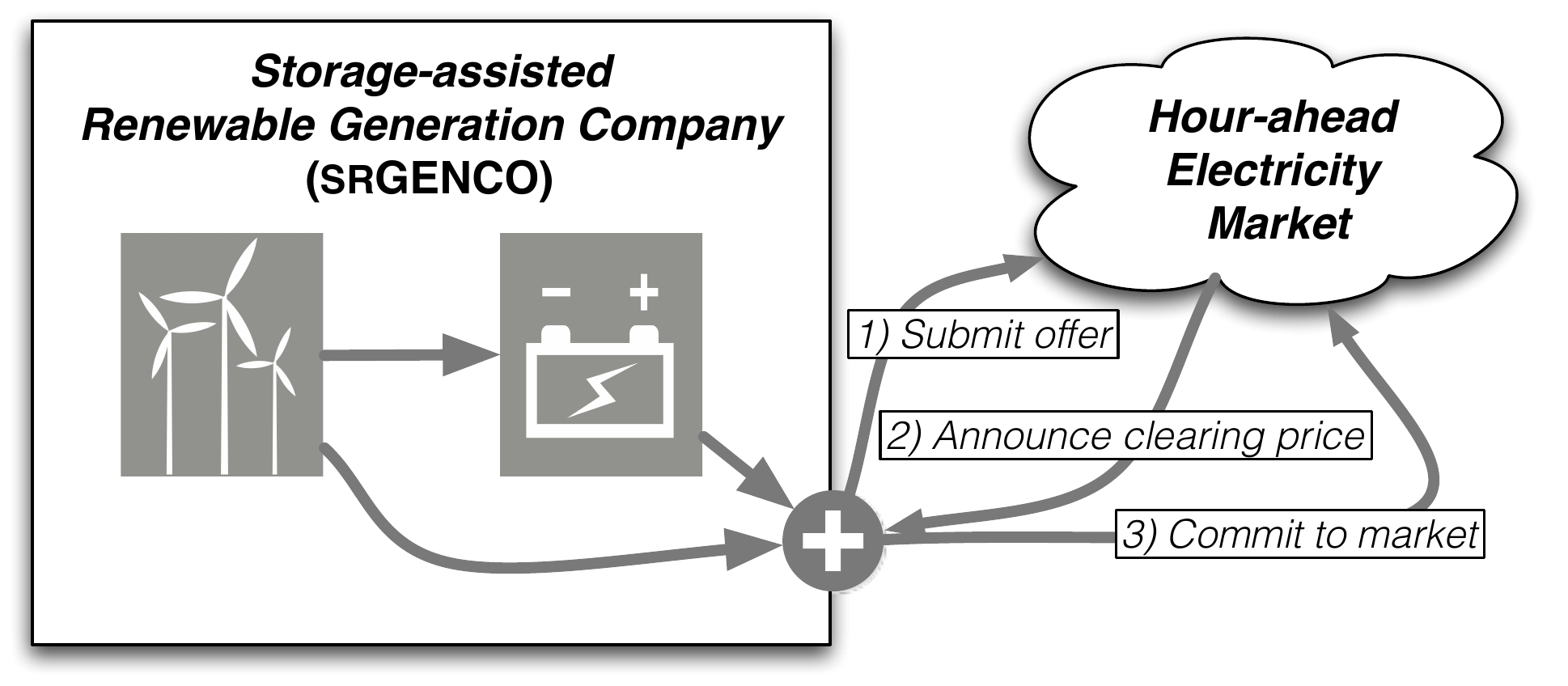}
	\end{center}
	\vspace{-5mm}
		\caption{The scenario}
		\label{fig:scenario}
\end{figure}



Finding profit maximization offering strategy for a renewable producer without storage, is nontrivial due to the inherent uncertainty of the renewables and dynamics in the market clearing price. In the presence of storage, the offering strategy is even more challenging because of the additional design space enabled by the storage. More specifically, at first, \rsgenco can use the storage to compensate for the slots that the renewable output cannot fulfill the commitment. In addition, the storage can potentially offer another economic advantage. That is, it can shift the energy through absorbing the renewable output during low price periods, and then discharge it during the high price periods. In this way, designing profit maximization offering strategy in the presence of storage comes with wider design space than those without the storage and potentially can bring more profit for \rsgenco.


To design the offering strategy, this paper casts an optimization problem with the objective of maximizing the long-term profit of \rsgenco subject to the storage capacity constraints. The inputs to the problem, \textit{i.e.}, the renewable output and the clearing price, however, are unknown for \rsgenco when submitting offer. This emphasizes the need for online solution design which is challenging, since the problem is coupled across time due to the evolution of the storage. We note that some similar problems have been studied in literature using offline ~\cite{castronuovo2004optimization,garcia2008stochastic} and stochastic optimization approaches~\cite{kim2011optimal,jiang2015optimal}. We refer to Sec.~\ref{sec:prob} for details.

 By introducing the system model and problem formulation in Sec.~\ref{sec:sys}, this work tackles the profit maximization offering strategy problem following online competitive algorithm design~\cite{Borodin98} and makes the following contributions:

\begin{enumerate}
	\item In Sec.~\ref{sec:sol}, we study a basic setting where the exact values of renewable output and the clearing price for the next slot are known to \rsgenco before submitting the offer. Even under this idealized setting, solving the problem is still challenging since the input to problem is still unknown for future slots beyond the next one. We propose \socs, a simple online offering strategy, in which the offering volume is calculated through a piecewise exponential/constant function of the renewable output and the current storage level, as well as the clearing price of the next slot.
	Our analysis demonstrates that \socs achieves the best possible competitive ratio of $O( \log \theta)$, where $\theta$ is the ratio between the maximum and minimum clearing prices.
	\item In Sec.~\ref{sec:online_nostep1}, we study the case where the clearing price is unknown and propose \ocsmb. In \ocsmb, \rsgenco submits multiple offers, each of which conveys a portion of the total offering volume, at different offering prices. Our analysis shows that the competitive ratio of \ocsmb converges to the ratio of \socs as the number of offers grows.
    Moreover, in Sec.~\ref{sec:online_nostep2}, our approach is extended to the case where \rsgenco knows renewable output with forecasting error. We propose \mocsmb as the generalized offering strategy and characterize the competitive ratio as a function of the maximum forecasting error.
	\item In Sec.~\ref{sec:exp}, by extensive numerical experiments based on real-world traces, we show that our online offering strategies can achieve satisfactory performance as compared to the offline optimum. In addition, \mocsmb with $10\%$ forecasting error improves the profit of \rsgenco by $15\%$ as compared to the baseline scenario without storage. As notable observations, our experiments demonstrate that when the market clearing price is unknown, submitting $3$ offers is sufficient to achieve almost the same performance as the case that the market price is known. Moreover, forecasting error of more than $20\%$ significantly degrades the performance results. In summary, our observations demonstrate that, while the uncertainty in market price can be effectively handled by multiple offer submissions, accurate short-term renewable forecasting is vital for \mbox{\rsgenco} to obtain a desired profit.
\end{enumerate}
%
%

We conclude the paper in Sec.~\ref{sec:con}. Due to space limitation, all proofs are included in our technical report~\cite{tech_report}.

\section{Model and Problem Formulation}
\label{sec:sys}
\begin{table}
	\caption{Summary of key notations}
	\label{tbl:not}
	\vspace{-3mm}
	\begin{center}
		\begin{tabular}{|c L{6.2cm}|}
			\hline
			\textbf{Notation} & \textbf{Description} \\
			\hline \hline
			$t$ & Index of one-hour time slot \\
			
			$T$ & The number of time slots, $T\geq0$\\
			
			$\mathcal{T}$ & Set $\mathcal{T} = \{1, 2, \dots, T\}$\\
			\hline
			\hline
			$p(t)$ & Market clearing price at  $t$, $p_{\min} \leq p(t) \leq p_{\max}$ \\
			
			$\theta$ & The ratio between the maximum and the minimum clearing prices, $\theta = p_{\max} / p_{\min}$\\
			\hline
			\hline
			$u(t)$ & The renewable output of \rsgenco at $t$\\
			
			$C$ & The capacity of storage system\\
			$\rho_c$ &  The maximum charging rate of storage system\\
			$\rho_d$ &  The maximum discharging rate of storage system\\
			
			$z(t)$  & The storage level at the beginning of $t$, see Eq.~\eqref{eq:et}\\
%
%
			\hline
			\hline
			$\hat{p}(t)$ & \textbf{optimization variable}, offering price of \rsgenco at $t$\\
			
			$\hat{x}(t)$ & \textbf{optimization variable}, offering volume of \rsgenco at $t$ \\
			\hline
			\hline
			$x(t)$ & Commitment volume of \rsgenco at $t$, see Eq.~\eqref{eq:xt}\\
			
			$y(t)$ & Over-commitment volume of \rsgenco at $t$, see Eq.~\eqref{eq:yt}\\
			
			$R(t)$ & The net profit of \rsgenco at $t$, see Eq.~\eqref{eq:rt}\\
			\hline
		\end{tabular}
	\end{center}
\end{table}

\subsection{Hour-Ahead Electricity Market}
\label{sec:hour}
The hour-ahead market operates in hourly basis\footnote{We emphasize that our model works in short-term market in which the offers are submitted \textit{before} actual market operation in hourly or even shorter scales, \textit{e.g.}, $15$ or $5$ minutes ahead. For real examples we refer to California ISO~\cite{caiso} and Nord Pool Markets~\cite{nord}. For a survey of electricity markets, we refer to~\cite{Harvard2011Primer}.} and \mbox{\rsgenco} along with other generation companies submits its offer, including the offering price and the offering volume (see Sec.~\ref{sec:offer2} for details), for the forthcoming hour shortly before the operation time. The market operator (usually known as independent system operator, ISO) matches the offers collected from the generation companies with the received bids from the demand-side, \textit{e.g.}, utility companies. Then, using a well-established double auction mechanism~\cite{liu2007risk} it determines the market clearing price for the next hour. The generation companies with the offering prices less than the clearing price are successful and the offering volume of electricity is considered as their \textit{commitment} to be sold on the market at the clearing price. Thus, successful offers sell at price at least as high as what they offered. All the remaining offers fail since their offering prices are greater than the clearing price.


Formally, we consider a time-slotted model, such that the time horizon $T$ is chopped into multiple slots with equal length, \textit{e.g.}, $1$ hour,  each of which is indexed by $t$. Shortly before slot $t$, \rsgenco along with other participants submits its offer, for the next slot. The ISO determines the clearing price $p(t)$ shortly after the participants submit their offers and bids. We assume $p_{\min} \leq p(t) \leq p_{\max}$, and parameter $\theta$ is defined as the ratio between the maximum and the minimum clearing prices, \textit{i.e.}, ${\theta = p_{\max} / p_{\min}}$. We later use $\theta$ to analyze our algorithms.

\subsection{The properties of \rsgenco}
There is a Storage-assisted Renewable GENeration COmpany (\rsgenco) that produces electricity from the renewable sources such as wind farm or solar plant. At the same time \rsgenco is equipped with the storage systems to store the electricity for future commitment with potentially higher price. On the other hand, the storage could be discharged to compensate for the shortage of renewable output when the commitment to the market is beyond the renewable output.
\subsubsection{Renewable Output}
The renewable output of \mbox{\rsgenco} at slot $t$ is denoted by $u(t)$ and we do not rely on any specific stochastic model of $u(t)$.
In general, we assume that \rsgenco does not know the exact amount of $u(t)$ when submitting the offer. Note that $u(t)$ could be (i) directly committed to the market, (ii) be committed to the market partially while the residual is stored on the storage, or (iii) entirely be stored on the storage for future usage.

\subsubsection{Offering Strategy\label{sec:offer2}}
By offering strategy we mean the way that \rsgenco determines its offer that includes:
\begin{enumerate}[(i)]
	\item \textbf{Offering price} denoted as ${\hat{p}(t) \in [p_{\min}, p_{\max}]}$, \textit{i.e.}, the minimum price at which \rsgenco desires to commit electricity to the market.
	\item \textbf{Offering volume} denoted as $\hat{x}(t)\geq 0$, \textit{i.e.}, the amount of electricity in MWh at which \rsgenco offers to the market at slot $t$.\footnote{Note that there is no upper bound for $\hat{x}(t)$ because we assume that the market is big enough to absorb the offering volume of \rsgenco entirely.}
\end{enumerate}
We distinguish between the offering volume and the \textit{commitment volume}. After the clearing price $p(t)$ is revealed, \rsgenco's offer may or may not be successful. If the offer is successful the offering volume is considered as the commitment volume. Otherwise, there would be no commitment. More specifically, we define $x(t)$ as the commitment volume of \rsgenco at slot $t$ as
\begin{equation}
\label{eq:xt}
x(t)=\left \{ \begin{array}{ll}
\hat{x}(t) & \qquad \textrm{if $p(t) \geq \hat{p}(t)$},\\
0 & \qquad \textrm{otherwise}.\\	
\end{array} \right.
\end{equation}
%

The goal of this study is to design an offering strategy for \rsgenco to submit both offering price $\hat{p}(t)$ and offering volume $\hat{x}(t)$ so as to maximize its long-term profit.

\subsubsection{Storage Model}
We denote the maximum capacity of storage system of \rsgenco by $C$ and let $\rho_c$ and $\rho_d$ be its maximum charging and discharge rates, respectively.
In addition, let $z(t)\in [0,C]$ be the storage level at the \textit{beginning} of slot $t$.
 Given the renewable output $u(t)$ and the commitment volume $x(t)$, the evolution of the storage level of \rsgenco is given by
\begin{equation}
\label{eq:et}
z(t+1)= \Big[z(t) + x_c(t) - x_d(t) \Big]_{\mathcal{C}},
\end{equation}
where
\begin{equation}
\label{eq:xc}
x_c(t) = \min\Big\{\rho_c,\big[u(t)-x(t)\big]^+\Big\},
\end{equation}
is the charging amount of the storage at slot $t$,
\begin{equation}
\label{eq:xd}
x_d(t) = \min\Big\{\rho_d,\big[x(t)-u(t)\big]^+\Big\},
\end{equation}
and $x_d(t)$ is the discharging amount of the storage at slot $t$.
 Moreover, $[.]^+$ and $[.]_{\mathcal{C}}$ define the projections onto the positive orthant and set $\mathcal{C} = [0,C]$, respectively.
Since \rsgenco is empowered by the storage, the commitment volume $x(t)$ might be either greater, less, or equal to the renewable output $u(t)$. The evolution of the storage for each case is as follows:
	
	(i) $u(t) = x(t)$: in this case the entire renewable output is committed to the market and there is no change on the storage level, \textit{i.e.}, $z(t+1)=z(t)$.
	
	(ii) $u(t) > x(t)$: in this case, $u(t)-x(t) > 0$ represents the amount of the surplus in the renewable output. Ideally, this surplus must be charged into the storage for the forthcoming commitments. However, because of the charging rate $\rho_c$ it may not be possible, which is indeed captured in Eq.~\eqref{eq:xc}.
	
	(iii) $u(t) < x(t)$: in this case not only the entire renewable output is committed, but also the storage should contribute in fulfilling the residual commitment, \textit{i.e.}, $x(t)-u(t) >0$. Again, given the storage level $z(t)$, and the discharge rate $\rho_d$, the residual commitment may not be satisfied.

\subsubsection{Over-Commitment}
Recall that \rsgenco does not know the exact renewable output $u(t)$ when submitting the offer, hence, it may fail to fulfill its commitment, which we refer to as the over-commitment. Let us denote $y(t)$ as the over-commitment volume at $t$ expressed as
\begin{equation}
\label{eq:yt}
y(t) = \bigg[x(t)-\Big(u(t)+\min \big\{z(t),\rho_d\big\}\Big)\bigg]^+.
\end{equation}
Note that the maximum amount that \rsgenco can provide to the market in operation time $t$ is the aggregation of the renewable output $u(t)$ and the maximum amount that could be discharged from the storage, \textit{i.e.}, $\min \big\{z(t),\rho_d\big\}$. Since $u(t)$ is unknown to \rsgenco when submitting the offer, the commitment volume $x(t)$ might be greater than the amount that \rsgenco can really output, \textit{i.e.}, $u(t)+\min \big\{z(t),\rho_d\big\}$.

\subsection{Profit Model}	
By committing $x(t)$, the profit obtained by \rsgenco is $p(t)x(t)$.
 The consequence of over-commitment is captured in profit model by augmenting a penalty term. We adopt the penalty model in \cite{kim2011optimal}, in which the unit penalty payment in over-commitment is linearly proportional to the spot price $p(t)$ in the form of $\alpha_1p(t)+\alpha_2$,
where $\alpha_1, \alpha_2\geq 0$ are constants.

Concluding above, the net profit obtained by \rsgenco at slot $t$, denoted by $R(t)$, is expressed as total profit subtracted by the (potential) penalty of the over-commitment, \textit{i.e.},
\begin{equation}
\label{eq:rt}
R(t)= p(t)x(t)-(\alpha_1 p(t)+\alpha_2)y(t).
\end{equation}




\subsection{Profit Maximization Problem}
\label{sec:prob}
The objective is to maximize the cumulative profit obtained by \rsgenco over time horizon $\mathcal{T}$. The profit maximization offering strategy problem (\csp) is formally casted as
\begin{eqnarray*}
\label{eq:problem}
 \csp & \textrm{max}&\sum\limits_{t\in\mathcal{T}} R(t),
\quad \textrm{s.t. Eq.}~\eqref{eq:et}, \\
&\textrm{vars.}& \quad \hat{p}(t) \in [p_{\min}, p_{\max}], \hat{x}(t) \geq 0, t\in\mathcal{T}.
\end{eqnarray*}
In offline scenario, in which the values of $u(t)$ and $p(t)$ as the time-varying inputs to the problem are known ahead of time, the problem is a linear one which is easy to solve. We refer to \cite{castronuovo2004optimization,garcia2008stochastic} as the related works that study related problems in offline settings. We note that in the offline scenario, the clearing prices are known to
\rsgenco, hence it is not require to submit the offering price anymore. Consequently, the offering strategy reduces to announcing the commitment volume directly. In this way, the problem could be reformulated in an equivalent form with simpler structure.

 In real-world, however,  neither the renewable output $u(t)$ nor the clearing price $p(t)$ are revealed to \rsgenco when submitting the offer. Hence the focus in this paper is to tackle the problem in online setting, so, we formulate the problem in a way that \rsgenco submits both offering price and offering volume.
Solving \csp in online setting is challenging, since the problem is coupled over the time in the presence of the storage system. Recall that an important advantage of incorporating storage towards profit maximization is to (fully or partially) store the renewable supply in the storage when that market price is low, and discharge it when the market price is high. Without knowing the future values of $p(t)$ and $u(t)$, finding a profit maximization offering strategy, that implicitly determines how renewable supply and the stored electricity in the storage must be consumed is challenging.



Finally, we note that in~\cite{kim2011optimal,jiang2015optimal}, following Markov decision process and approximate dynamic programming, different offering strategies are proposed given a particular probabilistic model of the clearing price and the renewable output. In these approaches, the solution is obtained in the sense of probabilistic expectation. In practice, however, real values might deviate from the underlying probabilistic models. Our general approach as explained in Sec.~\ref{sec:ocad} has no assumptions on the stochastic modeling of the unknown time-varying inputs.

\subsection{Our Approach: Online Competitive Algorithm Design}
\label{sec:ocad}
Our approach in this study is to follow competitive online algorithm design and propose online offering strategies in which the decision is made based on \textit{only the current information}, and without any assumptions on the stochastic model on the renewable output and the clearing price.
For the performance analysis, we use  \textit{competitive ratio}~\cite{Borodin98} which is a well-established metric to evaluate how good is the online solution.

\begin{myDef}
	When the underlying problem is a profit maximization one, for an online algorithm $\mathsf{A}$, its competitive ratio (\textsc{cr}) is defined as the maximum ratio between offline optimum and the profit obtained by $\mathsf{A}$, over all inputs, \textit{i.e.,}
	\begin{equation}
	\label{eq:cr}
	\cratio(\mathsf{A})\triangleq \max_{\boldsymbol{\omega} \in \Omega}\frac{R_{\emph{\ofa}}(\boldsymbol{\omega})}{R_{\mathsf{A}}(\boldsymbol{\omega})},
	\end{equation}
	where $\boldsymbol{\omega} \in \Omega$ refers to an instance of the online input parameters as
	\begin{equation}
	\label{eq:omega}
	\boldsymbol{\omega} \triangleq \big[\omega(t) = (u(t), p(t))\big]_{t\in\mathcal{T}},
	\end{equation}
	and $\Omega$ is the set of all input instances.
	Moreover, $R_{\emph{\ofa}}(\boldsymbol{\omega})$ and $R_{\mathsf{A}}(\boldsymbol{\omega})$ are the profits earned by the optimal offline solution and the online algorithm $\mathsf{A}$ respectively, when the input is $\boldsymbol{\omega}$.
\end{myDef}
By this definition, the smaller the competitive ratio, the better the performance is, since it guarantees no matter what the input is, the online optimal strategy obtains the profit close to the offline optimum.

\section{Optimal Online Offering Strategy with Accurate Single-Slot Prediction}
\label{sec:sol}
In this section, we propose online competitive algorithms for a simplified version of \csp, in which the accurate data for the next slot is available for both the renewable supply and the clearing price.
Later in Sec.~\ref{sec:online_nostep}, based on the insights from result of this section on the simplified scenario, we tackle the general case and propose online algorithms with neither the renewable supply nor the clearing price known to \rsgenco when submitting the bids.


\subsection{Simplified Problem  with Accurate Single-Slot Prediction}
\label{sec:online_onestep}
For the sake of simplification in design, we first assume that both $p(t)$ and $u(t)$ values for the next slot are revealed to \rsgenco, perhaps by accurate short-term forecasting tools. In this way, \csp is largely simplified in two ways:
\begin{enumerate}[(i)]
\item since $p(t)$ is known, the offering strategy reduces to finding just the commitment amount $x(t)$. We relax this assumption in Sec.\ref{sec:online_nostep1}.
\item since $u(t)$ is known, all the inputs and variables in Eq.~\eqref{eq:yt} are known to \rsgenco when submitting the offer, hence the over-commitment never happens, thereby the penalty term in the objective of \csp vanishes. We relax this assumption in Sec.~\ref{sec:online_nostep2}.
\end{enumerate}

Since in the new setting the price $p(t)$ is known for the next slot, and to be consistent with the general formulation, we set $\hat{p}(t) = p(t)$ and following Eq.~\eqref{eq:xt}, we get ${x(t) = \hat{x}(t)}$. Then, given the above assumptions, the only optimization variable would be the commitment amount $x(t)$. Now, we cast the simplified offering strategy problem \scsp as
\begin{eqnarray*}
\label{eq:scsp}
\scsp & \max &  \sum_{t\in\mathcal{T}} p(t)x(t) \\
&\mathrm{s.t.}& x(t)\leq \min \{ z(t),\rho_D\}+u(t), \\
&&z(t+1)=\Big[z(t)+ x_c(t)-x_d(t) \Big]_{\mathcal{C}},\\
&\mathrm{var}:& x(t)\geq 0, t\in\mathcal{T},
\end{eqnarray*}
where the first constraint ensures that over-commitment never happens. And the second constrain involves the evolution of the storage level, where $x_c(t)$ and $x_d(t)$ defined in Eqs.~\eqref{eq:xc}-\eqref{eq:xd} represent the charging and discharging amounts at time $t$.

\textbf{Remarks.} (1) Even though \scsp is simplified, it is still challenging and non-trivial in online setting. That is, for \scsp the forthcoming values for two inputs $p(\tau)$ and $u(\tau)$ for ${\tau \in \{t+1,\dots,T\}}$ are unknown, thereby the problem is still online, and again it cannot be decomposed across the time because of the evolution of the storage.

(2) \scsp is closely related to the classical \textit{time series search} and \textit{one-way trading} problems~\cite{Yaniv01,lorenz2009optimal}.\footnote{In one-way trading problem, a trader needs to exchange from one currency to another currency, given a time-varying exchange rates arriving online. The trader can decide to accept the current price or wait for the more attractive prices in future. The time series search problem is also quite similar~\cite{Yaniv01,lorenz2009optimal}.} As compared to these classical problems, \scsp can be characterized as a generalized version in two aspects: (i) \scsp introduces another exogenous parameter $u(t)$ in addition to the clearing price $p(t)$, this means that the adversary has more freedom to choose the worst-case input, this makes the competitive analysis more challenging; (ii) \scsp introduces the capacity constraint associated with the storage system. Putting together both differences, our investigations (see Sec.~\ref{sec:function_design}) demonstrate that neither \textit{fixed-threshold} policy, \textit{i.e.}, the policy that finds a fixed exchange rate and will accept the rates above that, nor the competitive analysis approach proposed in~\cite{Yaniv01} 
yields the optimal solution for \scsp in this paper. Instead, the optimal strategy for \scsp admits an \textit{adaptive-threshold-based} policy that achieves the optimal competitive ratio as a logarithmic function of $\theta$, whose details are stated in the next section.

\subsection{Online Algorithm Design for \scsp}
\label{sec:online_alg}
In this section, we propose a simple online offering strategy (\socs) to solve \scsp. Then we analyze its competitiveness and show that \socs achieves the best competitive ratio.

\subsubsection{High Level Intuitions}
\label{sec:alg_highlevel}
Intuitively, a proper offering strategy must consider two issues in decision making:
\begin{enumerate}[(i)]
	\item the clearing price $p(t)$ for the incoming slot $t$, the higher the price, the more the \rsgenco is willing to commit,
	\item the storage level $z(t)$, if the storage level is almost full, \rsgenco would be more interested to commit to have more capacity for the forthcoming slots to store the electricity. On the other hand, if the storage level is almost unoccupied, \rsgenco might keep this electricity to commit with higher price in future slots.
\end{enumerate}
Putting together both clearing price and the storage level, we design our algorithm following an \textit{adaptive threshold-based} strategy, since it adaptively changes the offering volume based on the current storage level and the clearing price.

\subsubsection{Main Algorithm}
The main idea is to construct a function $g(z):  [0,C] \rightarrow [p_{\min}, p_{\max}]$.\footnote{To be consistent to the notations in this paper, ideally we must denote the function as $g(z(t))$. However, for simplicity we drop the slot index $t$ for $z(t)$.} The input to function $g(\cdot)$ is the aggregation of the incoming renewable supply $u(t)$ and the current storage level $z(t)$, projected into the capacity of the storage, \textit{i.e.}, $\min \{z(t)+u(t), C\}$. The output of $g(\cdot)$ is the \textit{candidate offering price} $\hat{p}(t)$ for \rsgenco at slot $t$. Note that since we assume that \rsgenco knows clearing price $p(t)$, it calculates its candidate offering price $\hat{p}(t)$ and then based on the comparison between these two values decides how to submit its offer.

The function $g(\cdot)$ should be decreasing, \textit{i.e.}, with the increase in the current storage level, the offering price would be decreased (see intuition (ii) above).
Given function $g(z)$,
designing \socs is rather straightforward and is summarized as Algorithm~\ref{alg:1}.
At slot $t$, first, \socs calculates candidate offering price $\hat{p}(t)$ based on the given $g(z)$.
Then, it calculates the offering volume, whose details is explained in Sec.~\ref{sec:offer}.
In Line~\ref{algline:5}, the offering price is set to $\hat{p}(t) = p(t)$ to ensure that $\hat{x}(t)$ is committed to market anyway. In next, we explain how to design function $g(z)$.
\begin{algorithm}[!t]
	
	\caption{---\socs $\big[p(t),u(t), z(t)\big]$}
	\label{alg:1}
	\begin{algorithmic}[1]
		\State $z^+(t) \leftarrow \min \{z(t)+u(t), C\}$
		\State $\hat{p}(t) \leftarrow g(z^+(t))$; see Eqs.~\eqref{eq:gz1} and \eqref{eq:gz} for the optimal design of $g(z)$
		\State calculate $\hat{x}(t)$ according to $\hat{p}(t)$, see Eq.~\eqref{eq:hz} and the following analysis for the optimal design \label{algline:7}
		\State $\hat{p}(t) \leftarrow p(t)$\label{algline:5}
		\State submit offer $\big(\hat{p}(t),\hat{x}(t)\big)$;
	\end{algorithmic}
\end{algorithm}

\begin{figure}[!t]
	\centering
	\begin{minipage}{0.95\textwidth}
    \center
		\subfigure[]{
			\label{fig:a}
			\includegraphics[width=0.3\textwidth]{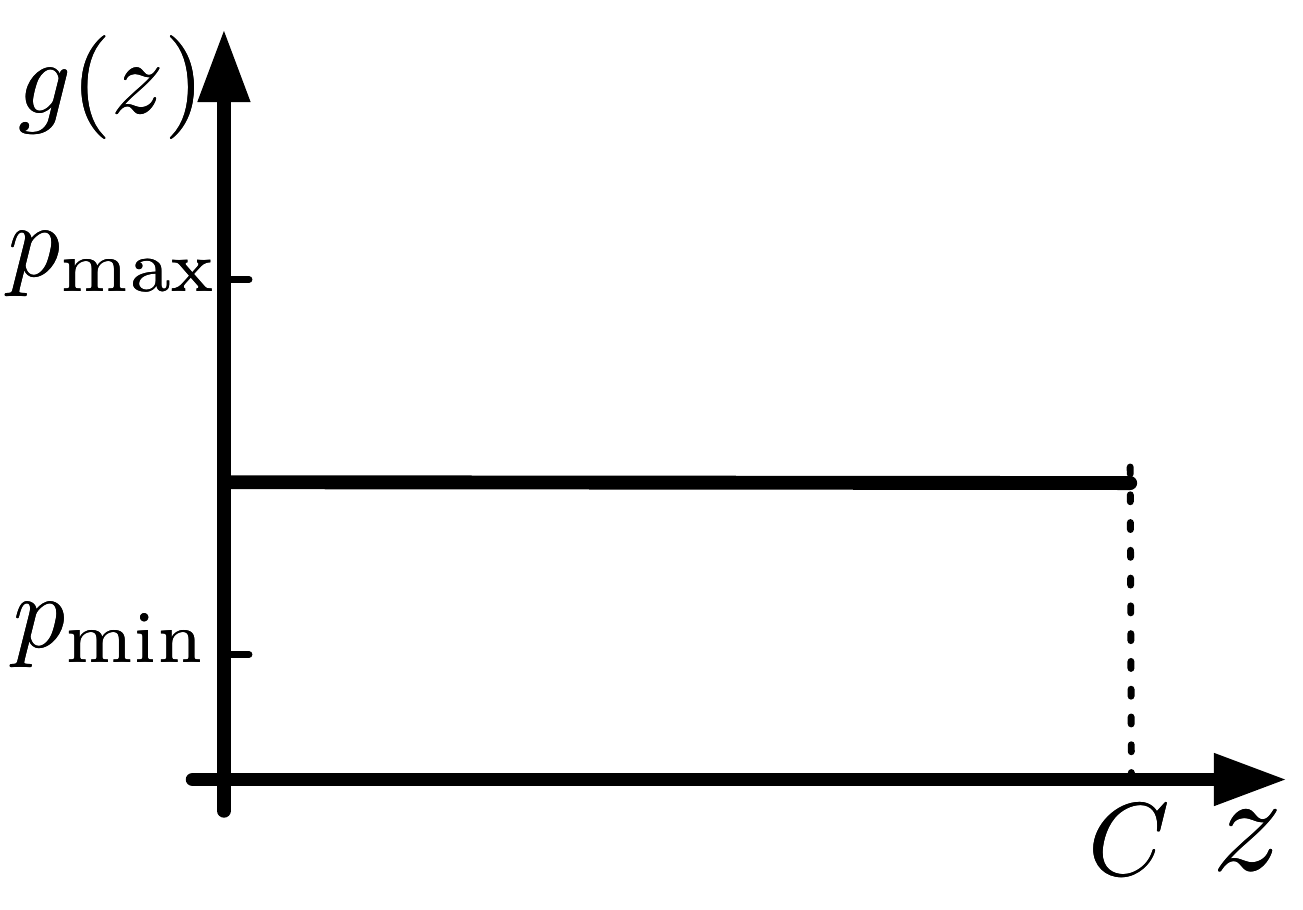}
		}
		\hspace{-5mm}
		\subfigure[]{
			\label{fig:b}
			\includegraphics[width=0.3\textwidth]{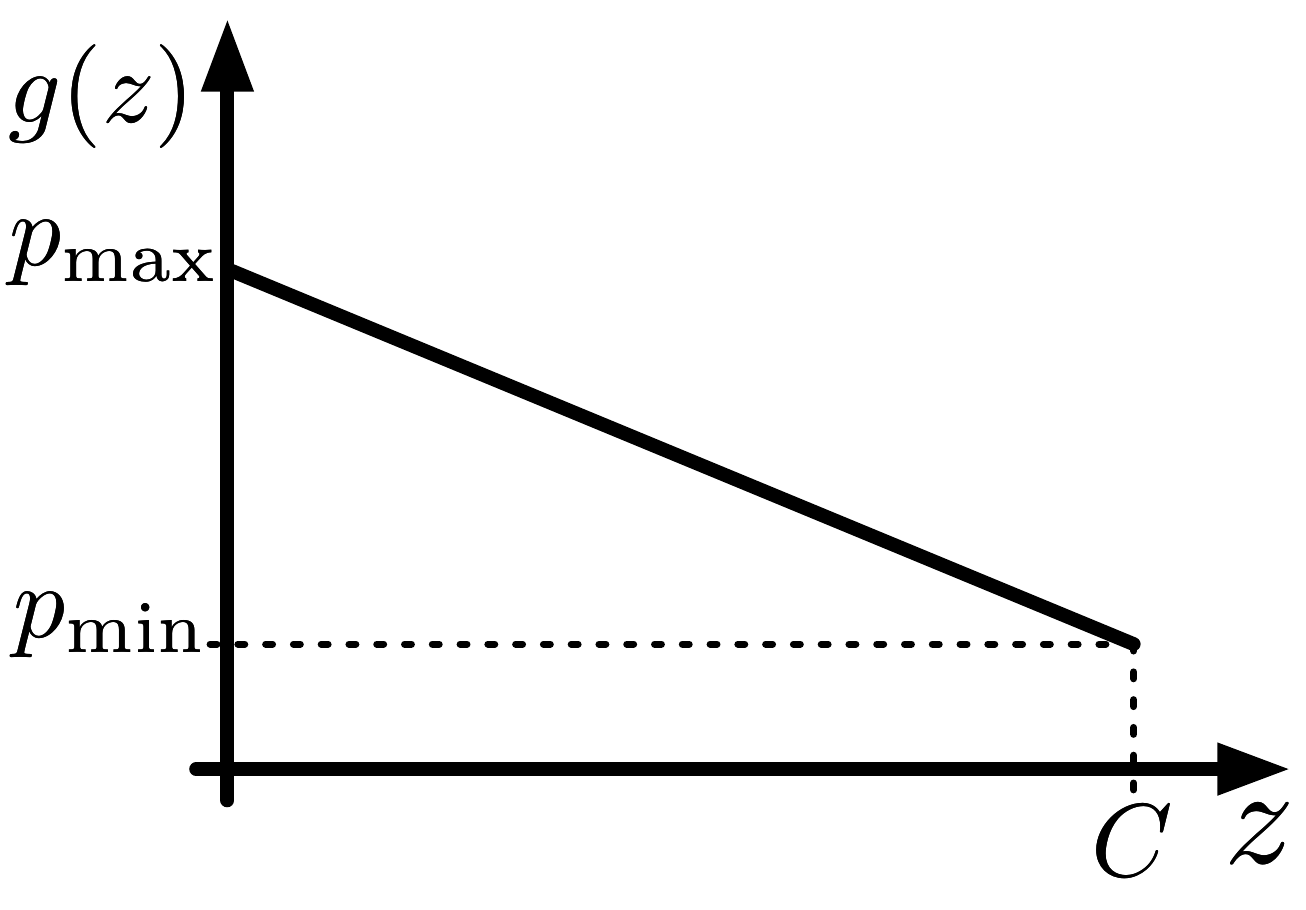}
		}
		\hspace{-5mm}
		\subfigure[]{
			\label{fig:c}
			\includegraphics[width=0.3\textwidth]{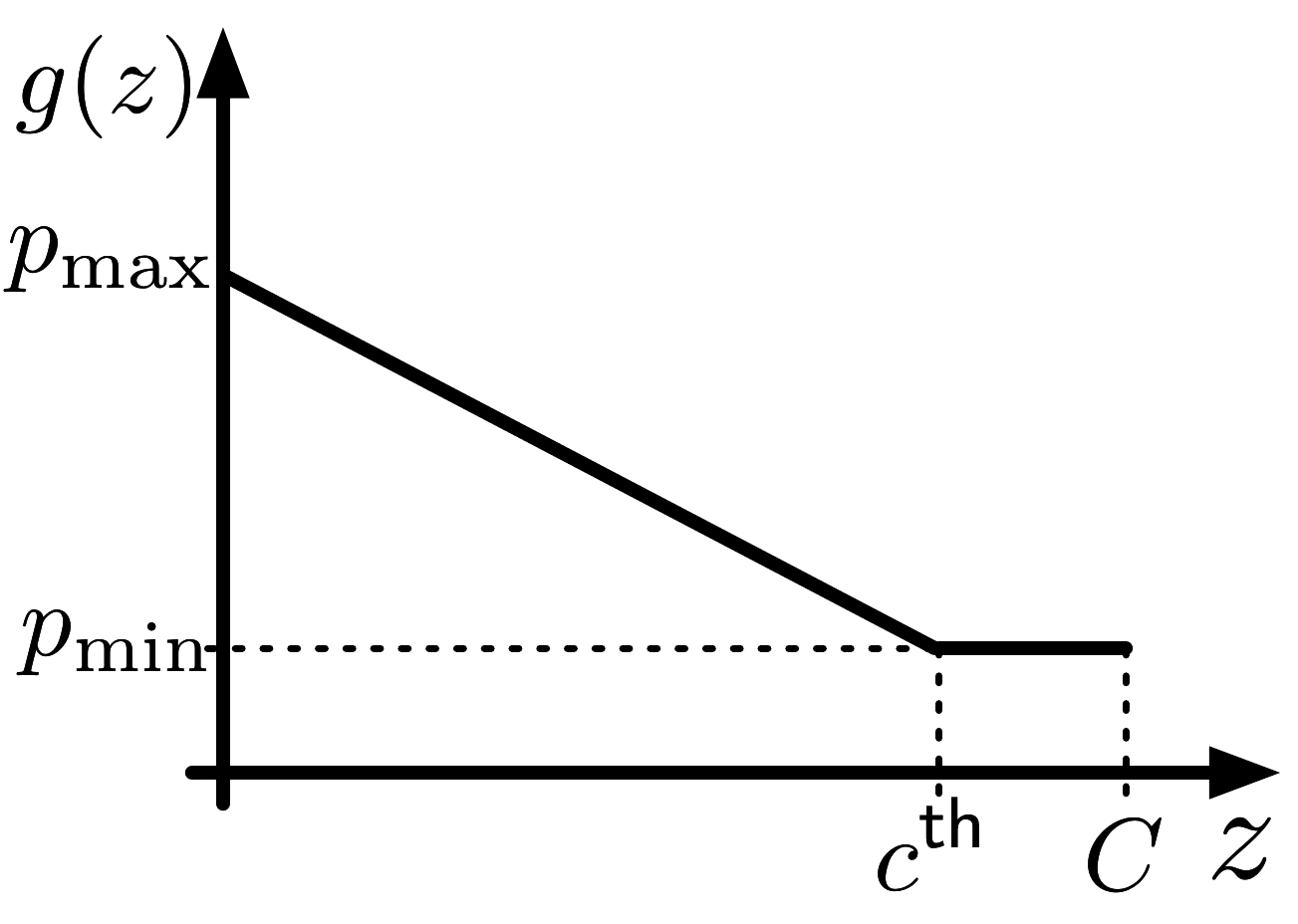}
		}
	\end{minipage}
	\vspace{-4mm}
	\caption{Different structures of function $g(z)$.}
	\label{figs:function_form}
\end{figure}

\subsection{The Design of Function $g(z)$}
\subsubsection{On the Importance of Designing Function $g(z)$}
\label{sec:function_design}
First, we highlight that $g(z)$ plays a critical role in \socs, and the competitive ratio can be improved by optimizing function $g(z)$.
To illustrate the impact of $g(z)$, we investigate the behavior of \socs under different structures of function  $g(z)$.

(i) $g(z) = c$, where $c \in [p_{\min}, p_{\max}]$, is a constant value as shown in Fig.~\ref{fig:a}. If $c >p_{\min}$, for the extreme case where $p(t)<c, t\in\mathcal{T}$, it could be easily shown that the profit obtained by the \socs is always $0$, thereby the competitive ratio would be $\infty$. On the other hand, $c = p_{\min}$, the competitive ratio is upper bounded by $p_{\max}/p_{\min}$, which is not intriguing, since \rsgenco always commits the electricity to the market regardless of the clearing price and, it loses the opportunity of utilizing the potentials of storage in obtaining more profit.

Consequently, an intelligent function design aims to improve the competitiveness by reserving the storage for forthcoming slots with higher clearing price. This goal can be achieved by adopting a decreasing function.

(ii) $g(z) =   a z + p_{\max}$, where $a < 0$, as depicted in Fig.~\ref{fig:b}. Again, under the extreme case, ${}{\cratio(\socs) \rightarrow
	\infty}$ if  ${p(t) = p_{\min}}$. Note that this linear structure for $g(z)$ is one example and any strictly decreasing function with $g(0)=p_{\max}$ and $g(C)=p_{\min}$ behaves similarly in the worst case.

(iii) Another smart alternative is a piece-wise function as depicted in Fig.~\ref{fig:c}. Function $g(z)$ again is decreasing initially. However, after the storage level reaches to a threshold value $c^{\mathsf{th}}$, $g(z)$ changes to a constant function, \textit{i.e.},
\begin{equation}
\label{eq:gz1}
g(z)=\left \{ \begin{array}{ll}
\hat{g}(z) & \qquad \textrm{if $z \le c^{\mathsf{th}}$},\\
p_{\min} & \qquad z\geq c^{\mathsf{th}}.\\	
\end{array} \right.
\end{equation}

Now, finding the optimal function $g(z)$ reduces to finding $\hat{g}(z)$ and $c^{\mathsf{th}}$. In the next, we introduce this function. Moreover in Sec.~\ref{sec:ca}, we analyze the competitiveness of the algorithm and prove that the proposed function achieves the optimal \cratio.
\subsubsection{Optimal Design of Function $g(z)$}
\label{sec:optimla_function_design}

The following theorem summarizes our main contribution in this section.

\begin{theorem}
	\label{thm:cr}
	By setting $\hat{g}(z)$ in Eq.~\eqref{eq:gz1} as
	\begin{equation}
	\label{eq:gz}
	\hat{g}(z) = p_{\min}e^{\frac{(c^{\mathsf{th}}-z)c^{\mathsf{th}}}{C(C-c^{\mathsf{th}})}},
	\end{equation}
	and $c^{\mathsf{th}}$ as
	\begin{equation}
	\label{eq:cth0}
	c^{\mathsf{th}} = C- \frac{(2+\log\theta)C-\sqrt{\log^{2}\theta+4\log\theta}C}{2} > 0,
	\end{equation}
	\emph{\socs} achieves the optimal competitive ratio for \emph{\scsp} as
	\begin{equation}
	\label{eq:cr_closed_form}
	\cratio(\emph{\socs})=\frac{(2+\log\theta)+\sqrt{\log^{2}\theta+4\log\theta}}{2}.
	\end{equation}
\end{theorem}
The proof is given as the competitive analysis in Sec.~\ref{sec:ca}.

\textbf{Remark.} The theorem shows that the competitive ratio is proportional to a logarithmic function of $\theta$ as the ratio between the maximum and minimum clearing prices. In practice, the scale of $\theta$ varies from $2$ to $50$ in different markets, \textit{e.g.}, the clearing prices in PJM~\cite{PJM} and NYISO~\cite{NYISO} in July, 2015 is in $[\$13.9,\$186.9]$ and $[\$8.1,\$43.1]$ per MWh. Given $\theta=50$, the competitive ratio is around $5.74$. Our experimental results depict much lower empirical ratios using the real prices in different markets. For details we refer to Table~\ref{tbl:comp}.


\begin{figure}[!t]
	\begin{center}
		\includegraphics[width=0.25\textwidth]{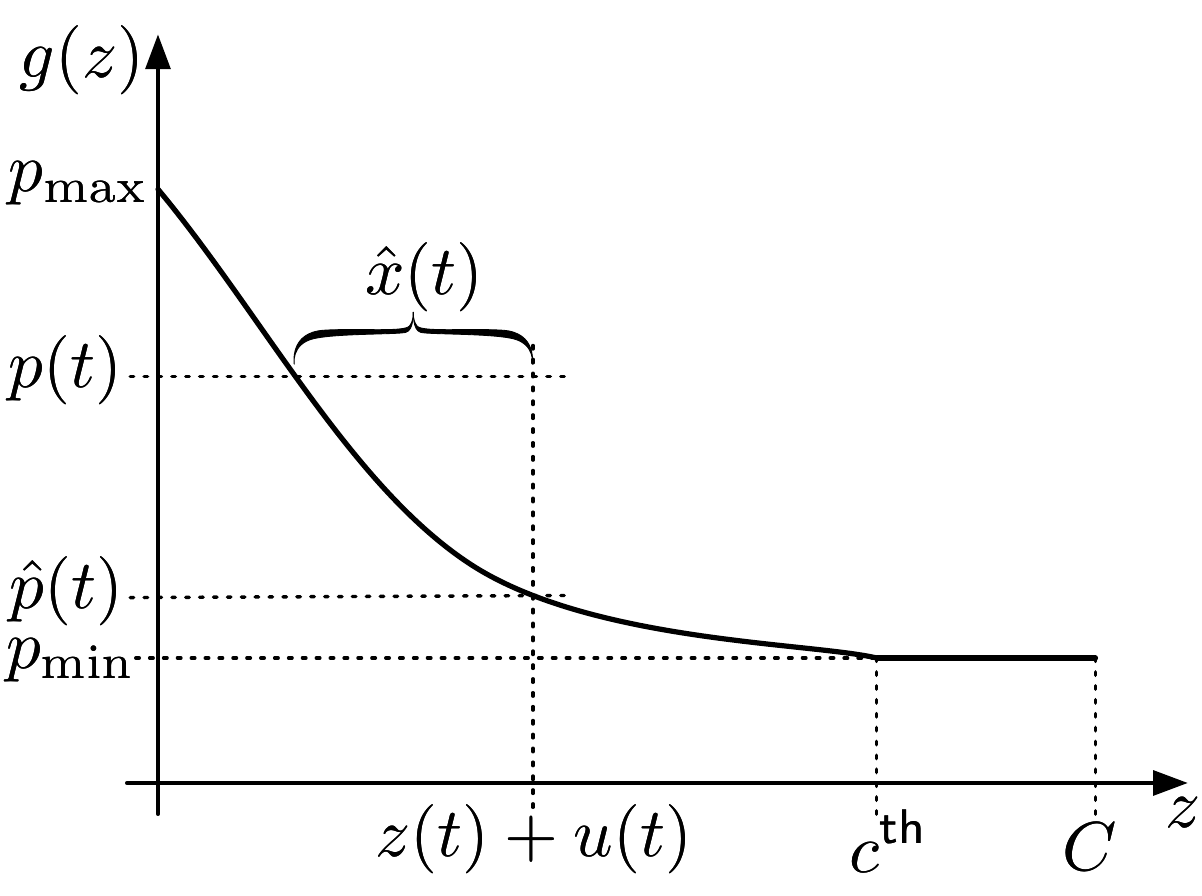}
	\end{center}
	\vspace{-4mm}
			\caption{Illustration of calculating the offering volume $\hat{x}(t)$ when $p(t)>\hat{p}(t)$. }
\label{fig:_bidding_amount}
\end{figure}

\subsubsection{Calculating the Offering Volume}
\label{sec:offer}
Given the optimal function $g(z)$ in Theorem~\ref{thm:cr}, we can finally find the offering volume $\hat{x}(t)$ as follows.
\begin{equation}
\begin{split}
\label{eq:hz}
&\hat{x}(t)= \\
&\!\!\left \{ \begin{array}{ll}
\!\![\rho_c - u(t)]^+, &  \!\!\!\!\textrm{if $\hat{p}(t) \!>\! p(t)$},\\
\!\!z(t)\!+\!u(t) \!-\! \min\{c^{\mathsf{th}},z(t)\!+\!\rho_{c}\}, &  \!\!\!\!\textrm{if $\hat{p}(t) \!=\! p(t) \!=\! p_{\min}$},\\
\!\!z(t) \!+\! u(t) \!- \!\min\{\hat{g}^{-1}(p(t)),z(t)\!+\!\rho_{c}\}, &  \!\!\!\!\textrm{if $\hat{p}(t) \!<\! p(t)$},\\
\end{array} \right.
\end{split}
\end{equation}
where $\hat{g}^{-1}(p(t))$ is well-defined since $\hat{g}(z)$ defined in Eq.~\eqref{eq:gz1} is monotonically decreasing. More specifically, when ${\hat{p}(t) > p(t)}$, \rsgenco stores the renewable supply as much as it can, considering the charging rate constraint $\rho_c$, then, the residual is committed to the market. If $\hat{p}(t) = p(t) = p_{\min}$, it means that the storage level exceeds the threshold level $c^{\mathsf{th}}$, so, \rsgenco offers the minimum price $p_{\min}$. Since the market price is also $p_{\min}$ it commits the electricity until the storage level reaches the threshold $c^{\mathsf{th}}$ or $z(t)\!+\!\rho_{c}$, \textit{i.e.}, $\hat{x}(t) = z(t)\!+\!u(t) \!-\! \min\{c^{\mathsf{th}},z(t)\!+\!\rho_{c}\}$. Finally, the last situation is $\hat{p}(t) < p(t)$, and $p(t) > p_{\min}$. In this case, we are in the exponential part of function  $g(z)$ and $\hat{x}(t)$ is calculated as the total supply $u(t) + z(t)$ subtracted by $\min\{\hat{g}^{-1}(p(t)),z(t)\!+\!\rho_{c}\}$. An illustration of the offering volume in this case when $\hat{g}^{-1}(p(t))\leq z(t)\!+\!\rho_{c}$ is depicted in Fig.~\ref{fig:_bidding_amount}.
Finally, to capture the maximum discharge rate of the storage, it suffices to modify offering volume as $\hat{x}(t) =  \min \{ u(t) + \rho_d, \hat{x}(t)\}.$

\subsection{Competitive Analysis}
\label{sec:ca}
Our goal in this section is to design function $g(z)$ (especially $\hat{g}(z)$ and $c^{\mathsf{th}}$ in Eq.~\eqref{eq:gz1}) so as to minimize the competitive ratio. By doing so we prove the result in Theorem~\ref{thm:cr}.
To simplify our explanation in this section and without loss of generality, we quantize the energy-related variables $x(t)$ and $z(t)$ and the input parameters $u(t)$ and $C$ to take the integer values. More specifically, considering that $C$ is a real number, we can define discretized capacity $C_{\textsf{d}}$ as
\begin{equation}
\label{eq:cd}
C_{\textsf{d}} = C/\eta,
\end{equation}
where $\eta$ is the unit of electricity. In this way, $C_{\textsf{d}}$ belongs to integer numbers. To avoid notation complexity, however, in our analysis in this section, we abuse notation $C$ to denote the discretized $C_{\textsf{d}}$. Similarly, we abuse $x(t), z(t)$, and $u(t)$ as their discretized versions.
Finally, we assume that initially the storage is full.

An immediate consequence of the above discretization is that function $g(z)$ could be considered as a step function.
We assume that there are $n$ steps for $g(z)$, each of which indexed by $i$, and $i\in\mathcal{N} = \{1,2,\ldots,n\}$. Generally, the length of each step could be different, hence, let us denote by $l_{i}$ as the length of step $i$. We can characterize the threshold value $c^{\textsf{th}}$ defined in Eq.~\eqref{eq:gz1} as the storage capacity subtracted by the length of the last step, \textit{i.e.},
\begin{equation}
\label{eq:cth}
c^{\textsf{th}} = C - l_n.
\end{equation}
Moreover, we define ${b_{i}=\sum_{k=1}^{i}l_{k}}$ as the cumulative length until step $i$. Since $z \in [0,C]$ we get $C=b_{n}=\sum_{k=1}^{n}l_{k}$. Finally, we denote $p_{i}=g(b_{i})$ as the threshold price when the storage level is $b_i$ under function $g(\cdot)$. By this definition, we get $p_{1}=p_{\max}$ and  $p_{n}=p_{\min}$.

By $z^{\socs}_{\boldsymbol{\omega}}(t)$, we denote the storage level at slot $t$ under \socs, and a particular instance $\boldsymbol{\omega}$ as defined in Eq.~\eqref{eq:omega}.
Let us denote the minimum storage level that \socs reaches under $\boldsymbol{\omega}$ as $b \in\{0,1,\ldots,C\}$, \textit{i.e.}, $\min_{t\in\mathcal{T}}z_{\boldsymbol{\omega}}^{\socs}(t)=b$. Then, using this definition, we can partition the universal set of input instances $\Omega$ to multiple subsets as follows
\begin{equation*}
\Omega =\bigcup_{b\in \{0,1,\ldots,C\}}\Omega_{b}^{\socs},
\end{equation*}
where $\Omega_{b}^{\socs}\triangleq \left\{\boldsymbol{\omega} \in\Omega: \min_{t\in\mathcal{T}} z^{\socs}_{\boldsymbol{\omega}}(t)=b\right\}.$

In particular, subset $\Omega_{b}^{\socs}$ represents the coalition of all input instances that results in the minimum storage level $b$ upon executing the deterministic online algorithm $\socs$.

\begin{myDef}
	\label{def:lcr}
Define the local competitive ratio $\cratio_{b}({\emph{\socs}})$ of ${\emph{\socs}}$ under the subset of input instances $\Omega_{b}^{{\emph{\socs}}}$ as
\begin{equation}
\cratio_{b}({{\emph{\socs}}})\triangleq\max_{\boldsymbol{\omega}\in\Omega_{b}^{{\emph{\socs}}}}\frac{R_{\emph{\ofa}}(\boldsymbol{\omega})}{R_{{\emph{\socs}}}(\boldsymbol{\omega})}.
\end{equation}
\end{myDef}
Given Definition~\ref{def:lcr}, we can redefine $\cratio(\socs)$ as follows.
\begin{myDef}
Define \mbox{\cratio(\emph{\socs})} as the maximum of $\cratio_{b}({{\emph{\socs}}})$ over all subsets $\Omega_{b}^{{\emph{\socs}}},b\in \{0,1,\ldots,C\}$, \textit{i.e.},
\begin{equation}
\label{eq:gcr}
\cratio({{\emph{\socs}}}) = \max_{b\in\{0,1,\ldots,C\}} \cratio_{b}({{\emph{\socs}}}).
\end{equation}
\end{myDef}

The following lemma characterizes a closed-form of $\cratio_{b}(\socs)$ for $b\in \{b_{1},b_{2},\ldots,b_{n-1}\}$.
\begin{lemma}
	\label{lem:lcr}
For \emph{\socs}, if ${\cratio_{b_{i}}(\emph{\socs})\geq \cratio_{b_{k}}(\emph{\socs})},k\in\{i+1,i+2,\ldots,n-1\}$ and $i\leq n-1$, then we have:
\begin{equation}
\label{eq:lcr_closed_form}
\cratio_{b_{i}}(\emph{\socs})=\frac{p_{i}C+\sum\nolimits_{k=i+1}^{n-1}p_kl_k}{\sum\nolimits_{k=i+1}^{n}p_kl_k}.
\end{equation}
otherwise,
\begin{equation}
\label{eq:lcr_closed_form_1}
\cratio_{b_{i}}(\emph{\socs})\geq\frac{p_{i}C+\sum\nolimits_{k=i+1}^{n-1}p_kl_k}{\sum\nolimits_{k=i+1}^{n}p_kl_k}.
\end{equation}
\end{lemma}

Through the proof of \ref{lem:lcr}, we can justify that the competitive ratio takes maximum value only among the subsets with ${b = b_i}$, ${i\in\{1,2,\ldots,n-1\}}$,
\emph{i.e.}, $${\cratio(\socs) = \max_{b\in\{b_1,b_{2},\ldots,b_{n-1}\}} \cratio_{b}(\socs)}.$$

Using the result in Lemma \ref{lem:lcr}, we get the global competitive ratio of $\socs$ under universal set of instances as follows
\begin{equation}
\label{eq:function_optimization}
\cratio(\socs)=\max_{i \in \{1,2,\ldots,n-1\}} \frac{p_{i}C+\sum\nolimits_{k=i+1}^{n-1}p_{k}l_{k}}{\sum\nolimits_{k=i+1}^{n}p_{k}l_{k}}.
\end{equation}

Our goal is to achieve the minimum possible value for $\cratio(\socs)$. Our design space toward this goal is to find: (i) the optimal value for $p_i$ which directly characterizes function $g(z)$, recall that by definition $p_i = g(b_i)$, and (ii) $l_n$ as the length of the last step in function which characterizes the threshold level $c^{\textsf{th}}$, recall that we have  $c^{\textsf{th}} = C - l_n$.

The following lemma states that the minimum global competitive ratio is achieved when the value of local competitive ratios are all equal.
\begin{lemma}
	\label{lem:nes_con}
Given a fixed $l_{1},l_{2},\ldots,l_{n}$, $\cratio(\emph{\socs})$ minimizes only if the following expression holds:
\begin{equation*}
\frac{p_{n\!-\!1}C}{p_{n}l_{n}}=\frac{p_{n\!-\!2}C+p_{n\!-\!1}l_{n\!-\!1}}{p_{n}l_{n}+p_{n\!-\!1}l_{n\!-\!1}}=\cdots=\frac{p_{1}C+\sum\nolimits_{k=2}^{n-1}p_{k}l_{k}}{\sum\nolimits_{k=2}^{n}p_{k}l_{k}}.
\end{equation*}
\end{lemma}

%
%
%

Using the result in Lemma~\ref{lem:nes_con} and by straightforward calculations, we can express $l_{i}$ as
\begin{equation}
\label{eq:l_i}
l_{i}=\frac{p_{i-1}-p_{i}}{p_{i}}\frac{p_{n}Cl_{n}}{p_{n-1}C-p_{n}l_{n}}, i\in\{2,3,\ldots,n-1\}.
\end{equation}
Given $C = \sum_{i=1}^{n} l_i$, and combining with \eqref{eq:l_i}, we get
\begin{equation*}
C=l_{1}\!+\!\frac{p_{n}Cl_{n}}{p_{n\!-\!1}C\!-\!p_{n}l_{n}}\!\sum_{i=2}^{n-1}\frac{p_{i-1}-p_{i}}{p_{i}}\!+\!l_{n}.
\end{equation*}

\begin{lemma}
	\label{lem:pn}
When $C\rightarrow \infty$, the competitive ratio of \socs takes the minimum value only if $p_{n-1}\rightarrow p_{n}$.	
\end{lemma}
Recall that $C$ is the discretized version  of the original storage capacity, and $C \rightarrow \infty$ could be achieved if we choose sufficiently small unit of electricity $\eta$ as in Eq.~\eqref{eq:cd}.
 Given the results in Lemma~\ref{lem:nes_con} and Lemma~\ref{lem:pn}, we have
\begin{eqnarray}
\label{eq:cr_ln}
\cratio(\socs)&=&\frac{C}{l_{n}}
=\frac{l_{1}}{l_{n}}\!\!+\!\frac{C}{C\!-\!l_{n}}\!\sum_{i=2}^{n-1}\frac{p_{i-1}\!-\!p_{i}}{p_{i}}+1\nonumber\\
&\geq & \frac{C}{C\!-\!l_{n}}\!\log\theta+1.
\end{eqnarray}
It can be verified that the above equation achieves the minimum value when $n\rightarrow \infty$, $l_{1}=0$, and ${l_{i}=1, i\in \{2,\ldots,n-1\}}$. Now, if $z<C-l_{n}$, we have
\begin{eqnarray*}
C-z&=\sum\limits_{k=z}^{n}l_{k}=\frac{p_{n}Cl_{n}}{p_{n-1}C-p_{n}l_{n}}\!\sum_{k=z}^{n-1}\frac{p_{k-1}-p_{k}}{p_{k}}+l_{n} \\
&\approx\frac{Cl_{n}}{C-l_{n}}\int_{p_{\min}}^{p_{z-1}}\frac{1}{p}dp+l_{n} \\
&=\frac{Cl_{n}}{C-l_{n}}\log\frac{p_{z-1}}{p_{\min}}+l_{n},\nonumber
\end{eqnarray*}
where the second-to-the-last equality holds since the difference between $p_{i}$ and $p_{i-1}$ would be arbitrarily small when $n\rightarrow \infty$ and ${l_{i}=1, i\in \{2,\ldots,n-1\}}$. Note that when ${1<z<C-l_{n}}$ and ${g(z)=p_{z-1}}$, solving the above equation we get the following closed-form for $g(z)$
\begin{equation}
\label{eq:gzln}
\begin{split}
&g(z)=\left\{\begin{array}{ll}
p_{\min} & \text{if } z\geq C-l_n, \\
p_{\min}e^{\frac{(C-l_n-z)(C-l_n)}{Cl_n}}& \text{otherwise}.
\end{array}
\right.
\end{split}
\end{equation}

Moreover, according to Eq.~\eqref{eq:cr_ln}, we have
\begin{equation*}
l_{n}=\frac{C}{\cratio(\socs)}=\frac{C}{\frac{C}{C-l_{n}}\log\theta+1},
\end{equation*}
and the closed form for $l_n$ is
\begin{equation}
\label{eq:ln}
l_n = \frac{(2+\log\theta)C-\sqrt{\log^{2}\theta+4\log\theta}C}{2}.
\end{equation}

Putting together the results in Eqs.~\eqref{eq:gzln},~\eqref{eq:ln}, and~\eqref{eq:cth}, the result in Theorem~\ref{thm:cr} is proved.

\section{Online Offering Strategy without Accurate Single-Slot Prediction}
\label{sec:online_nostep}
In Sec.~\ref{sec:online_nostep1}, we first extend the previous result to the case that the clearing price $p(t)$ is unknown to \rsgenco when submitting the offer, however, the renewable output $u(t)$ is known accurately. Second, in Sec.~\ref{sec:online_nostep2}, we extend the result to the general case that the renewable output is known to \rsgenco with forecasting error.
\subsection{\emph{\ocsmb}: \emph{\socs} with multiple offer submissions; $p(t)$ Is Unknown, $u(t)$ Is Known}
\label{sec:online_nostep1}

Our general approach in this scenario is to use the potentials of submitting \textit{multiple offers} which is allowed in the current markets, \textit{e.g.}, at most 10 offers are permitted for \rsgenco in PJM market~\cite{PJM_supply_curve}.
Our approach in this case is to calculate total feasible commitment volume, and then partition this total amount into multiple offers, each of which conveying a portion of the offering volume, in different prices.
In Algorithm~\ref{alg:2}, we summarize the detail of \ocsmb, as the multiple version of \socs. In \ocsmb, we denote  $\mathcal{B}=\left\{(p_1,x_1),(p_2,x_2), \ldots,(p_m,x_m)\right\}$ as the set of multiple offers submitted by the \rsgenco, where $m$ is the maximum number of offers that is permitted by the market operator. In \ocsmb, if $\min\{u(t),\rho_{c}\}+z(t)>c^{\textsf{th}}$, it means that the current storage level $z(t)$ plus the chargeable renewable supply, \textit{i.e.,} $\min\{u(t),\rho_{c}\}$, exceeds the threshold value $c^{\textsf{th}}$, so, we can safely submit the surplus amount with the minimum price (Line~\ref{algline2:3}), and partition the remaining amount equally in $m-1$ offers (Lines~\ref{algline2:4}-\ref{algline2:8}). Otherwise, we submit the excessive energy $[u(t)-\rho_{c}]^+$ at price $p_{\min}$ and partition the residual equally into $m-1$ offers (Lines~\ref{algline2:10}-\ref{algline2:15}).
\begin{algorithm}[!t]
	\caption{---\ocsmb $\big[u(t),z(t)\big]$}
	\label{alg:2}
		\begin{algorithmic}[1]
			\State declare $\mathcal{B} \leftarrow \emptyset$ as the set of offers
			\If{$\min\{u(t),\rho_{c}\}+z(t)>c^{\textsf{th}}$}\label{algline2:2}
			\State$\mathcal{B}\leftarrow\left\{(\!u(t)\!+\!z(t)\!-\!c^{\textsf{th}},p_{\min}\!)\right\}$\label{algline2:3}
            \State $\Delta x\leftarrow\min\{c^{\textsf{th}},u(t)+\rho_d\}/(m-1)$\label{algline2:4}
			\For{$i \leftarrow1$ to $m-1$}
			\State $p_{i} \leftarrow g(c^{\textsf{th}}-i\Delta x)$
			\State $\mathcal{B}\leftarrow \mathcal{B} \cup \left\{(\Delta x,p_{i})\right\}$
			\EndFor\label{algline2:8}
			\Else
            \State$\mathcal{B}\leftarrow\left\{([u(t)-\rho_{c}]^+,p_{\min}\!)\right\}$\label{algline2:10}
			\State $\Delta x\leftarrow[u(t)+\min \{z(t),\rho_d\}]/(m-1)$
			\For{$i\leftarrow1$ to $m-1$}
			\State $p_{i} \leftarrow g(z(t)+u(t)-i\Delta x)$
			\State $\mathcal{B}\leftarrow \mathcal{B} \cup \left\{(\Delta x,p_{i})\right\}$
			\EndFor\label{algline2:15}
			\EndIf
			\State submit $\mathcal{B}$
		\end{algorithmic}
\end{algorithm}
Theorem~\ref{thm:2} characterizes the competitive ratio of \ocsmb as a function of \cratio(\socs) and the number of submitted offers $m$.

\begin{theorem}
	\label{thm:2}
	The competitive ratio of \emph{\ocsmb} is bounded by
	\begin{equation*}
	\cratio(\emph{\ocsmb})\leq \left(1+\frac{\cratio({\emph{\socs}})\theta}{m^2}\right)\cratio({\emph{\socs}}).
	\end{equation*}
\end{theorem}
Note that $\cratio(\ocsmb) \rightarrow \cratio(\socs)$ as $m \rightarrow \infty$. In experiments, we evaluate the impact of the number of offers on the performance of \ocsmb.

\subsection{\emph{\mocsmb}: generalized \emph{\ocsmb}; $u(t)$ Is Given with Forecasting Error, $p(t)$ Is Unknown}
\label{sec:online_nostep2}
Finally, we release the accurate forecasting assumption of the renewable output, and assume that the error of power generation in the forthcoming slot is bounded in a particular region. Namely the input to this algorithm is $\tilde{u}(t)$ as the predicted value and $e(t)$ as the maximum error, such that
\begin{equation*}
u(t)\in [(1-e(t))\tilde{u}(t),(1+e(t))\tilde{u}(t)],~0\leq e(t)\leq e_{\max}.
\end{equation*}
Our algorithm \mocsmb is a simple extension of \ocsmb, with replacing $u(t)$ with  $(1-e(t))\tilde{u}(t)$ as the minimum possible value for the renewable output. We skip the other details since they are the same as \ocsmb. Note that with this input, the algorithm behaves in the most conservative way, such that the over-commitment never happens. Extending the algorithm to the more aggressive cases that takes into account the risk of over-commitment is part of our future work.
In following theorem we characterizes the competitive ratio of \mocsmb.
\begin{theorem}
\label{thm:3}
	Assuming $e_{\max}<0.5$, the competitive ratio of \emph{\mocsmb} is bounded by
$
	\cratio(\emph{\mocsmb})\leq \frac{1}{1-2e_{\max}}\cratio(\emph{\ocsmb}).
$
\end{theorem}
The proof for Theorem \ref{thm:3} is very simple. It's obvious that the amount of committed energy by \mocsmb is always larger than $(1-2e_{\max})x_{\ocsmb}$ for any $t\in \mathcal{T}$.

Clearly, as $e_{\max} \rightarrow 0$, we get $	\cratio(\mocsmb) \rightarrow \cratio(\ocsmb)$. In experiments, we investigate the impact of forecasting error on the result of \mocsmb.

\section{Experimental Results}
\label{sec:exp}
In this section, we evaluate the performance of our online strategies using the real-word traces for the renewable output and electricity market prices. Our objective is two-fold:
(i) to compare the performance against the optimal offline, a comparison algorithm~\cite{lorenz2009optimal}, and a baseline in which there is no storage, and
(ii) to investigate the impact of the system model and algorithm parameters.
\subsection{Experimental Settings}
\subsubsection{Renewable Output and Electricity Market Prices} We use the wind generation data from PJM energy market~\cite{PJM} with the capacity of $10$MW. We note that this is a wind farm in moderate size and the world largest has an operational capacity of $1020$MW~\cite{WindCapacity}.
The hourly electricity price data are from PJM market for most of the experiments. We also demonstrated our results for Nord Pool~\cite{nord} and NYISO~\cite{NYISO} markets in Table~\ref{tbl:comp}. We note that the prices exhibit severe seasonal patterns. In particular, the prices are highly volatile, \textit{e.g.}, in summer the peak price reaches as high as $\$396.9/$MWh, in $2015$. For this reason we evaluate the performance of our algorithms in different seasons as well.

\begin{figure}[!t]
	\minipage{0.45\textwidth}
    \center
	\includegraphics[width=0.9\textwidth]{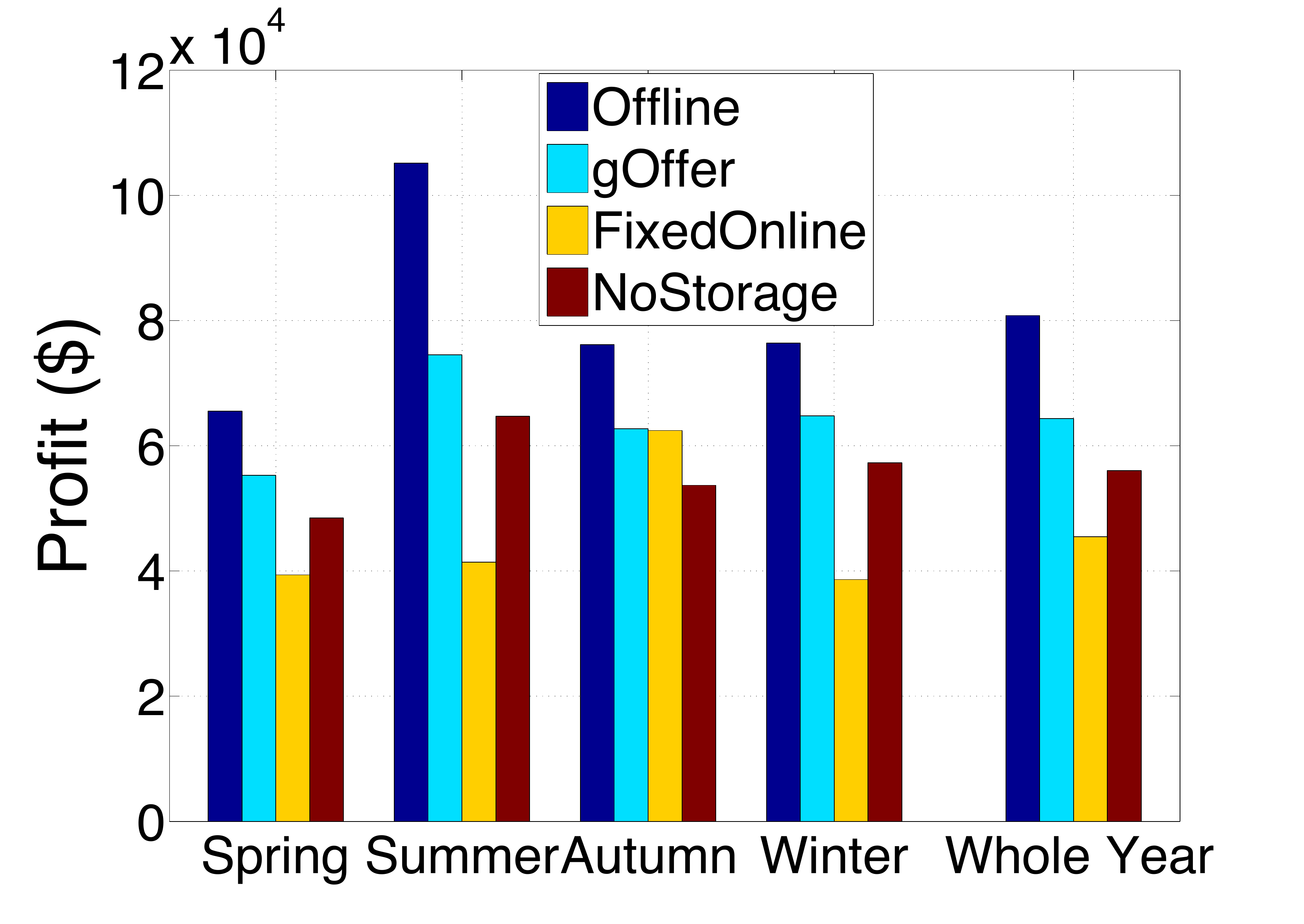}
	\vspace{-6mm}
	\caption{The performance of different algorithms in different seasons}
	\label{fig:SeasonalChange}
	\endminipage\hfill		
	\minipage{0.45\textwidth}
    \center
	\includegraphics[width=0.9\textwidth]{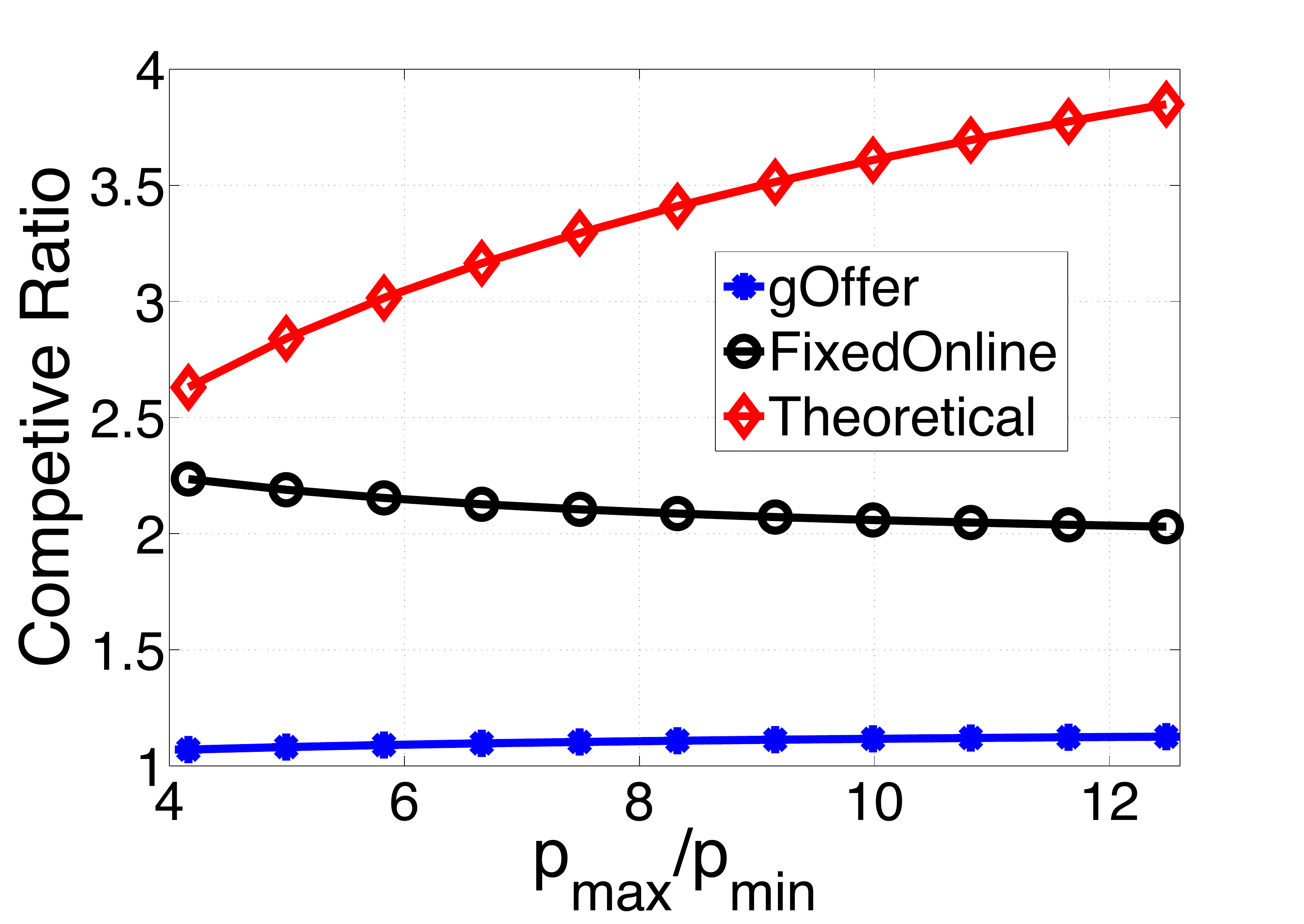}
	\vspace{-6mm}	\caption{The competitive ratio as a function of price volatility}
	\label{fig:PriceRange}
	\endminipage\hfill
\end{figure}

\subsubsection{Storage Capacity} Unless otherwise specified, the storage capacity is set to $20$MWh. The maximum charge and discharge rates are $10$MW. In reality,  large scale Compressed Air Energy Storage (CAES) with similar parameters has been developed to cope with renewable uncertainty~\cite{luo2015overview}.

\subsubsection{Parameters for the Algorithms} Unless otherwise specified, for \ocsmb and \mocsmb, the default value for the number of offers is $10$. This is the common practice in PJM market~\cite{PJM_supply_curve}. Moreover, the maximum forecasting error $e_{\max}$ in \mocsmb is set to $10\%$. Finally, we note that each data point in figures is the average results of $100$ different runs of the algorithms with $T=360$ hours.

\subsubsection{Comparison Algorithms} We compare the performance of our algorithm \mocsmb with three other alternatives: (1) \ofa, the optimal offline solution that is implemented as the benchmark to obtain the empirical competitive ratio; (2) \nostorage, the optimal offline cost when there is no storage. This is used to evaluate the economic advantage of integrating the storage; (3) \fonline, another simple online algorithm with a fixed threshold price. Specifically, we follow the approach in~\cite{lorenz2009optimal} and set the threshold of this simple online algorithm fixed at $\sqrt{p_{\min}p_{\max}}$, regardless of the storage level. In this algorithm, \rsgenco commits all the electricity whenever the price is not smaller than this threshold. The acronyms for all the algorithms are summarized in Table~\ref{tbl:comp_alg}.

\begin{table}
	\caption{Summary of comparison algorithms}\vspace{-5mm}
	\label{tbl:comp_alg}
	\begin{center}
		\begin{tabular}{|c L{6.8cm}|}
			\hline
			\textbf{Algorithm} & \textbf{Description} \\
			\hline
			\hline
			& \textbf{Our algorithms}\\
			\hline
			\socs & Simplified online offering strategy; $p(t)$ and $u(t)$ are known\\
			\ocsmb & \socs with multiple submissions; $p(t)$ is unknown, $u(t)$ is known\\
			\mocsmb & Generalized \ocsmb; $p(t)$ is unknown, $u(t)$ is known with error\\
			\hline
			\hline
			& \textbf{Comparison Algorithms}\\
			\hline
			\ofa & Optimal offline solution with storage\\
			\nostorage & Optimal offline solution without storage\\
			\fonline & Simple online algorithm with fixed threshold price~\cite{lorenz2009optimal}\\
			\hline
		\end{tabular}
	\end{center}
\end{table}

\subsection{Experimental Results}
\subsubsection{Comparison Results across Different Seasons}
In this experiment we report the profit obtained by different algorithms in different seasons as well as the whole year. The result is depicted in Fig.~\ref{fig:SeasonalChange}. The main observations are: (1) \mocsmb achieves $80\%$ of the offline optimum, which shows that it is close to optimal. (2) \mocsmb outperforms \nostorage by $15\%$, which signifies the substantial economic benefit of incorporating the storage. (3) \mocsmb outperforms \fonline by $42\%$, which depicts the superiority of our online algorithms as compared to other ``storage-level-oblivious'' online alternatives.

\begin{table}
	\caption{Summary of Theoretical and Empirical Competitive Ratios on Different Electricity Markets}
	\label{tbl:comp}
	\centering
	\begin{tabular}{|c|c|c|c|}
		\hline
		Electricity market & $\theta = p_{\max}/p_{\min}$  & Theoretical \cratio &  Empirical \cratio\\
		\hline \hline
		PJM & 13.44 & 4.37 & 1.18 \\
		\hline
		NYISO & 5.32 & 3.38 & 1.14 \\
		\hline
		Nord Pool & 3.63 & 2.95 & 1.09 \\
		\hline
	\end{tabular}
\end{table}

\subsubsection{Impact of the Price Volatility}
The electricity price in the deregulated electricity market exhibits large fluctuation. Theoretically, large price volatility will degrade the performance of the online algorithm, as the competitive ratio is an increasing function of $\theta= p_{\max}/p_{\min}$. In this experiment, we present the result under different values of $\theta$.
As shown in Fig.~\ref{fig:PriceRange}, \mocsmb is robust to price fluctuation with less than $5\%$ increment, even though the theoretical competitive ratio increases by $44\%$. Meanwhile, we note that the empirical competitive ratio of \fonline decreases slightly as the $\theta$ increases. However, it is on average $90\%$ larger than that of \mocsmb, which further signifies the superiority of \mocsmb.
In addition, we report the result of \mocsmb for the prices in different markets in Table~\ref{tbl:comp}. The result signifies that the larger the price volatility, the large theoretical and empirical competitive ratios.

\begin{figure}
	\minipage{0.45\textwidth}
    \center
	\includegraphics[width=0.9\textwidth]{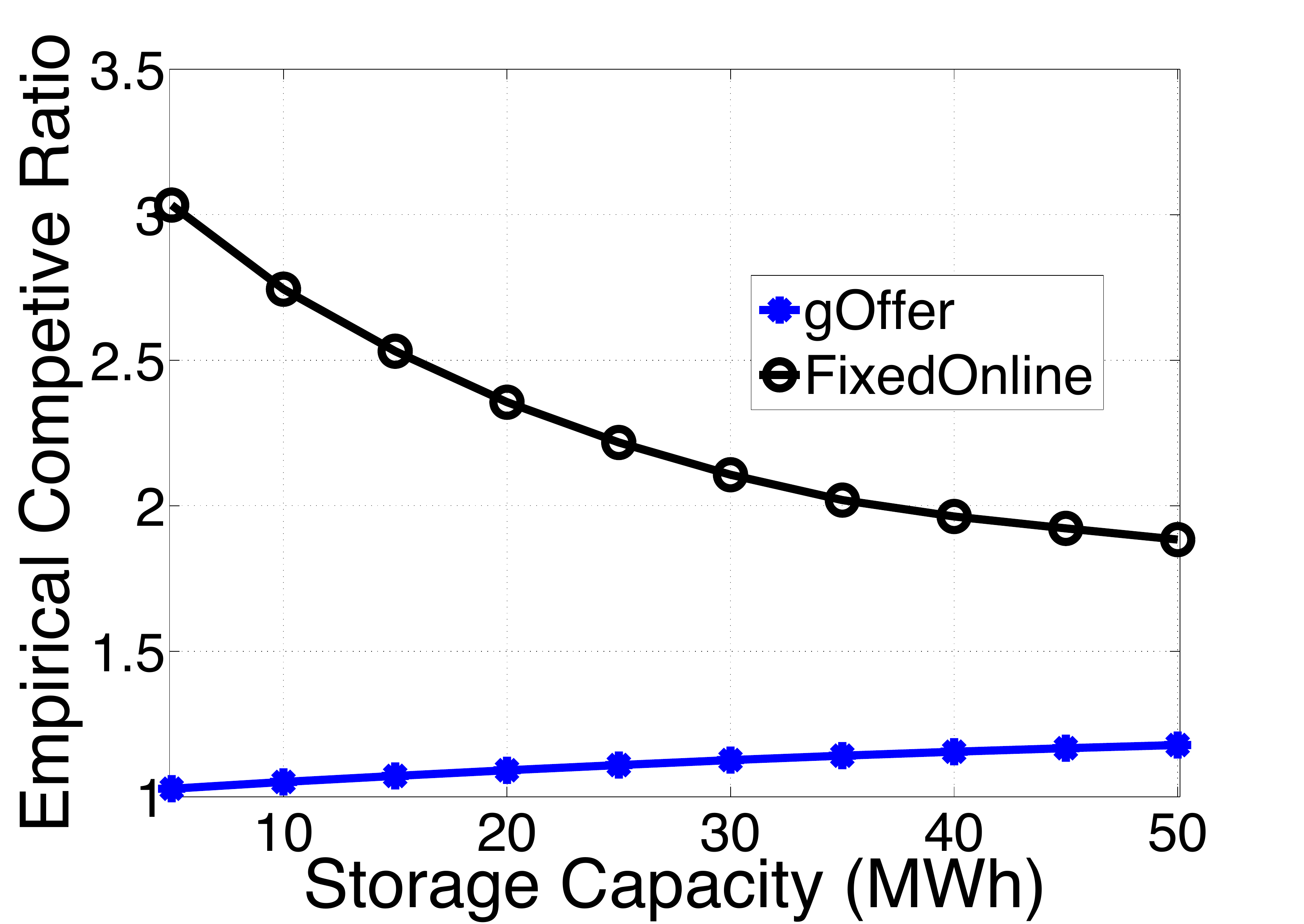}
	\vspace{-6mm}	\caption{The competitive ratio as a function of storage capacity}
	\label{fig:BatteryCapacity}
	\endminipage\hfill
	\minipage{0.45\textwidth}
    \center
	\includegraphics[width=0.9\textwidth]{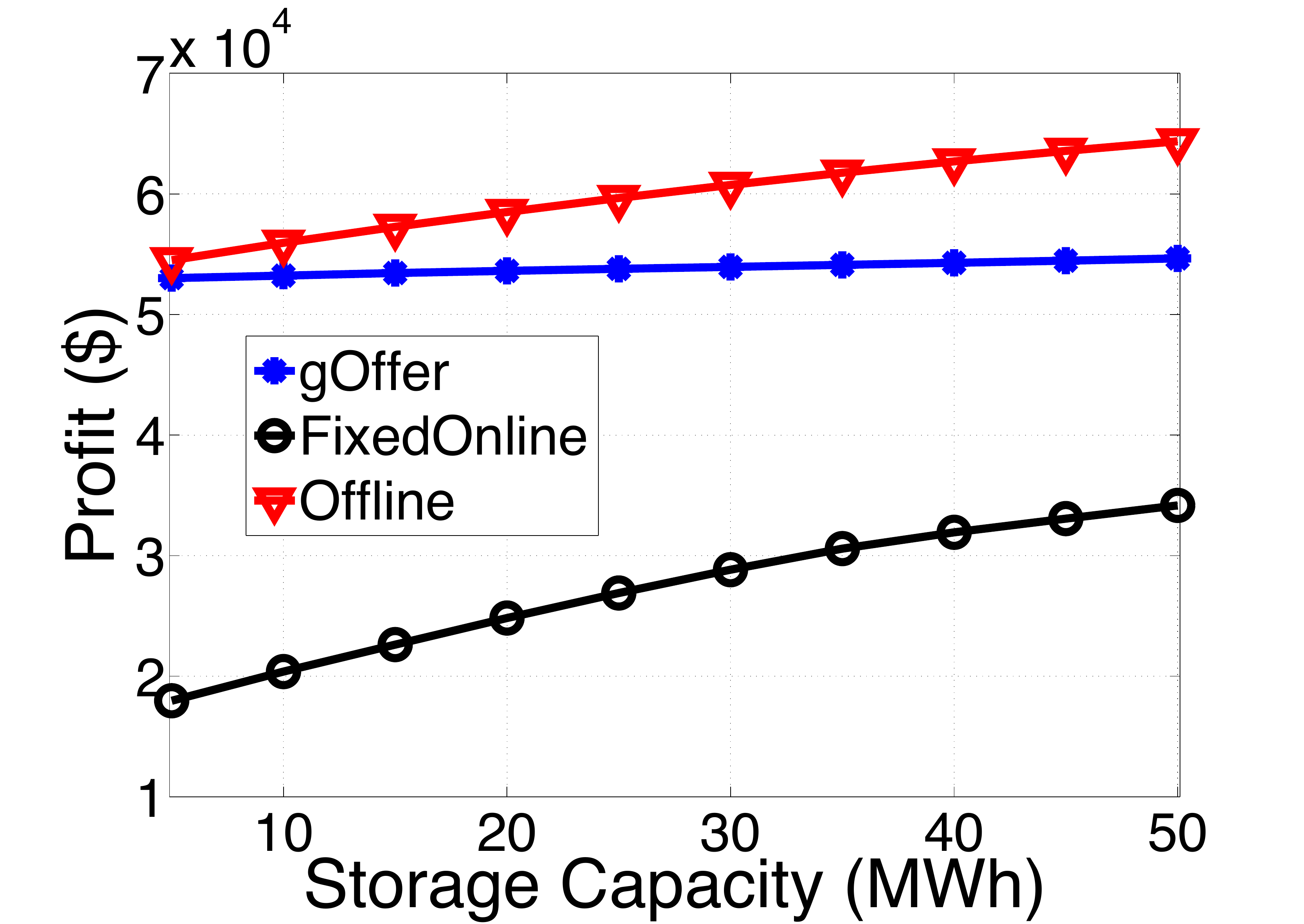}
	\vspace{-6mm}	\caption{The performance of different algorithms as a function of capacity}
	\label{fig:BatteryCapacity_profit}
	\endminipage\hfill
\end{figure}

\subsubsection{Impact of the Storage Capacity}
Storage capacity planning is an important issue that \rsgenco's owner needs to consider, since the storage is still expensive with the current technology. In this experiment, we vary the storage capacity from $5$ to $50$MWh to investigate its impact on the profit of \rsgenco.
Fig.~\ref{fig:BatteryCapacity} and Fig.~\ref{fig:BatteryCapacity_profit} show the empirical competitive ratios and the obtained profits, respectively. As the storage capacity increases, an increase in profit of both online algorithms is observed. However, the increase in \mocsmb is smaller ($3\%$) than that of \fonline ($90\%$). This is mainly because \fonline is completely oblivious to the storage level, and with the increase in capacity, there would be more room to mitigate this unawareness.  Meanwhile, the empirical competitive ratio of \mocsmb increases with large storage capacity (from $1.03$ to $1.18$).
This result depicts that when the capacity is in the order of the renewable capacity (say, $\times 0.5$ to $\times 2$), \mocsmb is close-to-optimal. However, when the storage capacity is much higher than the renewable capacity (say, $\times 5$), perhaps more sophisticated algorithms are required.

\subsubsection{Impact of the Uncertainty of the Clearing Price and the Renewable Output}
In the last set of experiments we investigate the impact of the number of offers in \ocsmb (in Fig.~\ref{fig_NumOfBids}) and the forecasting error in \mocsmb (in Fig.~\ref{fig:PreError}).
In \ocsmb, we relaxed the assumption of \socs and extend it to the case that the clearing price $p(t)$ is unknown. We proposed to submit multiple offers to alleviate its negative impact. To investigate how many offers are sufficient for \ocsmb to achieve the same performance level as \socs, in Fig.~\ref{fig_NumOfBids}, we vary the number of offers from $1$ to $15$. The notable observation is that submitting $1$ or $2$ offers is not sufficient. However, with $3$ or more offers the performance is quite similar to \socs in which the price is known in advance.
In the last experiment, we increase the maximum error of renewable output $e_{\max}$ and calculate the profit of \mocsmb. The result shows that \mocsmb is robust to forecasting error that belows $20\%$, and the obtained profit decreases rapidly as error increases beyond $20\%$.
Concluding above, these experiments demonstrate that the negative impact of the uncertainty in the clearing price can be effectively mitigated by multiple offer submissions. However, accurate short-term renewable forecasting is vital for \rsgenco to obtain a desired profit, since the errors higher than $20\%$ can severely degrade the performance.

\begin{figure}[!t]
	\minipage{0.45\textwidth}
    \center
	\includegraphics[width=0.9\textwidth]{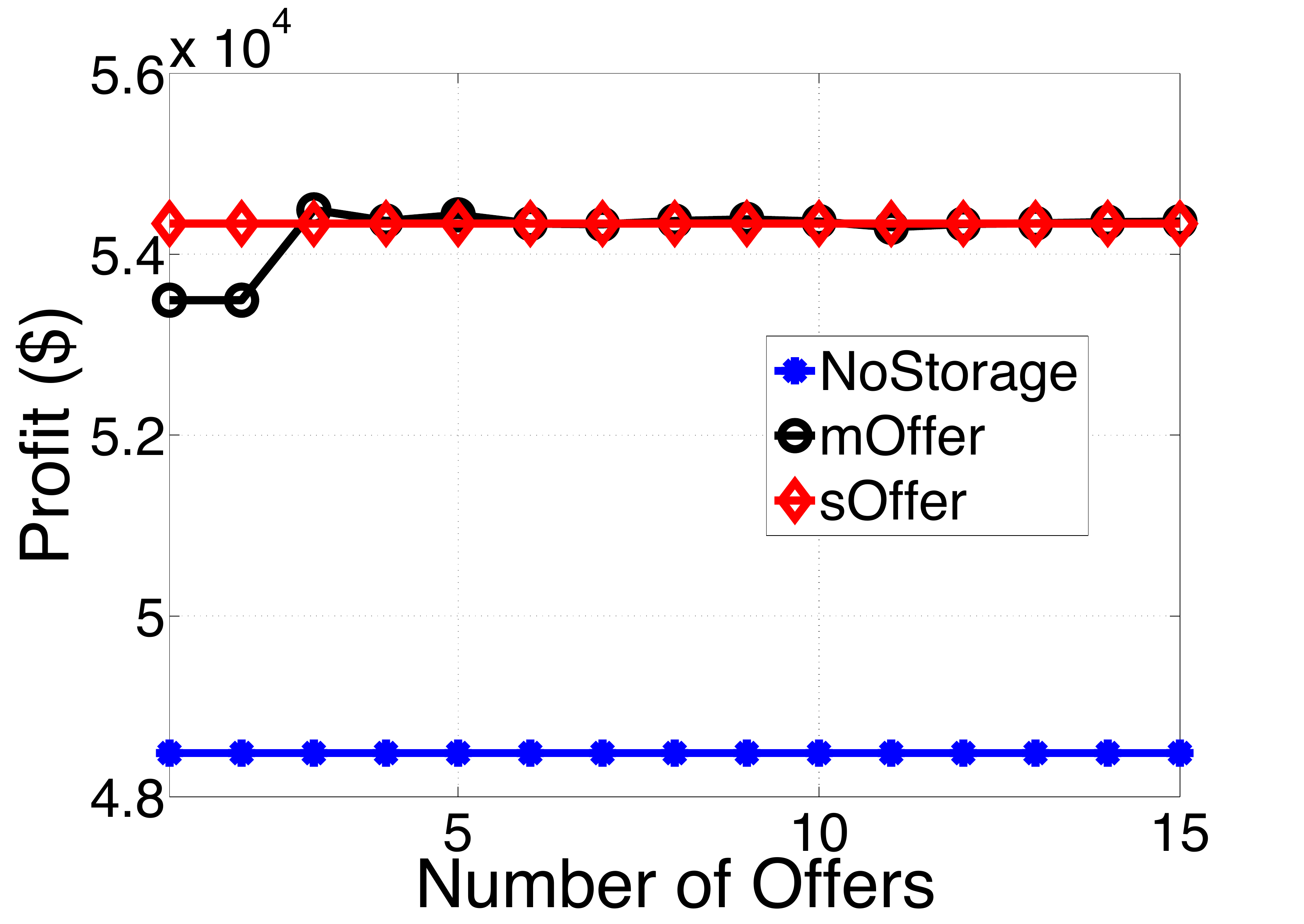}
	\vspace{-6mm}\caption{The performance of \ocsmb as the number of bids increases}
	\label{fig_NumOfBids}
	\endminipage\hfill		
	\minipage{0.45\textwidth}
    \center
	\includegraphics[width=0.9\textwidth]{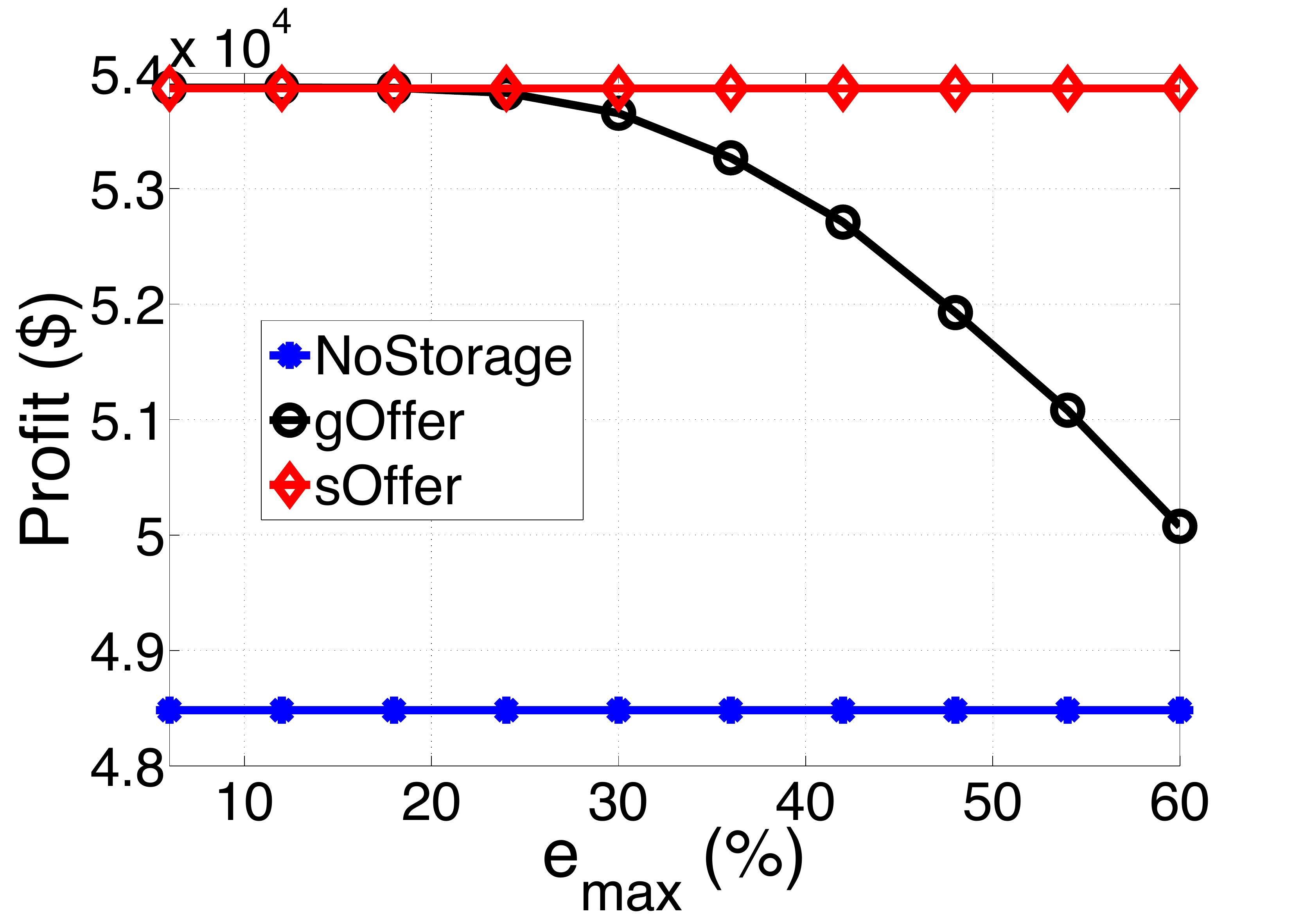}
	\vspace{-6mm}\caption{The performance of \mocsmb as forecasting error increases}
	\label{fig:PreError}
	\endminipage\hfill
\end{figure}

\section{Conclusion}
\label{sec:con}
This paper proposes competitive offering strategies for maximizing the profit of a renewable power producer with storage in hour-ahead market. The underlying problem is coupled over time and the input to the problem is revealed in an online fashion, both of which make the problem challenging. We first tackle a basic setting and propose algorithm for it. And in two successive steps we relax the basic assumptions and generalize our solutions and analysis to the realistic scenario.
Our analytical results characterize competitive ratios for all the proposed algorithms.
Finally, trace-driven evaluations demonstrate the close-to-optimal performance of our algorithms.


\bibliographystyle{abbrv}


\appendix
\label{sec:appendix}

\subsection{proof of Lemma \ref{lem:lcr}}
\begin{lemma}
\label{lem:mp}
Under some worst case $\omega^*$, assume ${\min_{t\in \mathcal{T}}z_{\omega^*}^{\socs}(t)=b}$ then we have ${p(t)\le g(b)+\sigma}$ for any ${\sigma > 0}$.
\end{lemma}
\begin{proof}
Assume there's a time slot that $p(t)\ge g(b)+\sigma$ under some worst case $\omega^*$. If $x_{\omega^*}^{\ofa}(t)\le x_{\omega^*}^{\socs}(t)$, we just construct another worst case which replaces the $t-\text{th}$ input $(p(t),~u(t))$ with the following input sequence:
\begin{equation*}
\begin{split}
&p'(t)=g(z_{\omega^*}^{\socs}(t)+u(t))+\sigma,~u'(t)=u(t), \\
&p'(t+1)=g(z_{\omega^*}^{\socs}(t+1))+\sigma,~u'(t+1)=0, \\
&p'(t+2)=g(z_{\omega^*}^{\socs}(t+2))+\sigma,~u'(t+2)=0, \\
&\vdots \\
&p'(t+N)=g(z_{\omega^*}^{\socs}(t)+u(t)-x_{\omega^*}^{\socs}(t)),~u'(t+N)=0
\end{split}
\end{equation*}
where $N=\lceil x_{\omega^*}^{\socs}(t)/\sigma\rceil$. In this way, the profit earned by \socs decreases while that of \ofa keep unchanged. This contradicts the assumption at the beginning. Then we complete the proof.
\end{proof}

We adopt a worst-case analysis approach to prove Lemma \ref{lem:lcr}. We assume the state of storage under online algorithm reaches the minimum value $b_{i}$ at time $T_{1}$ and keeps unchanged until time $T_{2}$.

If there's no injection into the storage during $[1,T_{2}]$, according to Lemma \ref{lem:mp}, the maximum market clearing price $[1,T_{1}]$ is $g(b_{i})$ and
the maximum profit earned by \ofa during $[1,T_{1}]$ is $g(b_{i})C$. The minimum profit earned by the online algorithm is $\sum_{k=b_{i}+1}^{b_{n}}g(k)$, so it's immediate that the maximum profit ratio is $g(b_{i})C/\sum_{k=b_{i}+1}^{b_{n}}g(k)$.

(a) If there's energy stored into the battery during $[1,T_{2}]$, we can assume the last time interval with energy stored is $T_{I}=[T_{3},T_{4}]\subset [1,T_{2}]$. During $T_{I}$, the state of storage under \socs will keep increasing or unchanged. Then the off-line algorithm will commit $X=\sum_{t\in T_{I}}u(t)$ amount of energy before $T_{3}$ with lower price than $g(b_{i})$. Moreover, we assume at time $T_{5}$, the state of storage under \socs reaches $z^{\omega^*}_{\socs}(T_{3})$ for the last time during $[1,T_{2}]$. Then the committed energy during $[T_{3},T_{5}]$ by \socs is equal to that by \ofa.

If
\begin{equation*}
\frac{g(z^{\omega^{*}}_{\socs}(T_{3})-1)X}{\sum\limits_{k=z^{\omega^{*}}_{\socs}(T_{3})+1
}^{z^{\omega^{*}}_{\socs}(T_{3})+X}g(k)}>\cratio_{b_{i}}(\socs)
\end{equation*}
we can easily get that:
\begin{equation*}
\begin{split}
\cratio_{\Omega^{\socs}_{z^{\omega^{*}}_{\socs}(T_{3})}}(\socs) & \geq  \frac{(g(z^{\omega^{*}}_{\socs}(T_{3}))+1)z^{\omega^{*}}_{\socs}(T_{3})}{\sum\limits_{k=b_{n}}^{x^{\omega^{*}}_{\socs}(T_{3})}g(k)} \\
&\geq\frac{(g(z^{\omega^{*}}_{\socs}(T_{3}))+1)X}{\sum\limits_{k=z^{\omega^{*}}_{\socs}(T_{3})+1
}^{z^{\omega^{*}}_{\socs}(T_{3})+X}g(k)} \\
&>\cratio_{\Omega^{\socs}_{b_{i}}}(\socs)
\end{split}
\end{equation*}

Thus, for the worst case, if the maximum profit ratio within the case set $\Omega^{\socs}_{b_i}$ satisfies $\cratio_{\Omega^{\socs}_{b_{i}}}(\socs)\geq \cratio_{\Omega^{\socs}_{b_{k}}}(\socs)$ for $k=i+1,i+2,\ldots,n-1$, there's no stored energy during $[1,T_{2}]$ under the worst case.

(b) If there'is energy committed by \socs during $[T_{3}+1,T]$ under the worst case, we can assume the period with energy committed during $[T_{3}+1,T]$ by \socs is $T_{D}=[T_{6},T_{7}]$. And assume the time that the state of battery first reaches $b_{0}=z_{\socs}^{*}(T_{7})$ at time $T_{8}$. It's easy to verify that, at time $T_{7}$ and $T_{8}$, the state of storage under \ofa is $0$. So the committed energy by \ofa algorithm during $[T_{3}+1,T_{7}]$, denoted by $X'$ is equal to that by \socs. During $T_{D}$, the profit earned by online algorithm is at least $\sum_{k=z^{\omega^{*}}_{\socs}(T_{6})-X'}^{z^{\omega^{*}}_{\socs}(T_{6})}g(k)$, and the maximum profit earned by the \ofa is $(g(z^{\omega^{*}}_{\socs}(T_{6}))+X')X'$.

If
\begin{equation*}
\frac{(g(z^{\omega^{*}}_{\socs}(T_{7}))-1)X'}{\sum\limits_{k=z^{\omega^{*}}_{\socs}(T_{6})-X'}^{z^{\omega^{*}}_{\socs}(T_{6})}g(k)}>\cratio_{\Omega^{\socs}_{b_{i}}}(\socs)
\end{equation*}
we can easily get that:
\begin{equation*}
\begin{split}
\cratio_{\Omega^{\socs}_{z^{\omega^{*}}_{\socs}(T_{6})}}(\socs)\geq&\frac{(g(z^{\omega^{*}}_{\socs}(T_{6}))-1)z^{\omega^{*}}_{\socs}(T_{6})}{\sum\limits_{k=b_{n}}^{z^{\omega^{*}}_{\socs}(T_{6})}g(k)} \\
\geq &\frac{(g(z^{\omega^{*}}_{\socs}(T_{6}))-1)X'}{\sum\limits_{k=z^{\omega^{*}}_{\socs}(T_{6})-X'}^{z^{\omega^{*}}_{\socs}(T_{6})}g(k)} \\
>&\cratio_{\Omega^{\socs}_{b_{i}}}(\socs)
\end{split}
\end{equation*}

Thus, for the worst case, if the maximum profit ratio within the case set $\Omega^{\socs}_{b_{i}}$ satisfies $\cratio_{\Omega^{\socs}_{b_{i}}}(\socs)\geq CR_{\Omega^{\socs}_{b_{k}}}(\socs)$ for $k=i+1,i+2,\ldots,n-1$, there's no energy committed during $[T_{3}+1,T]$ by \socs, and the maximum profit gained by the off-line algorithm is $\sum_{k=b_{i-1}+1}^{b_{n-1}}g(k)$.

Concluding (a) and (b), we complete the proof.

\subsection{proof of Lemma \ref{lem:nes_con}}
Given $l_{1},l_{2},\ldots,l_{n}$, solving \eqref{eq:function_optimization} can be formulated as the following problem:
\begin{equation*}
\label{function_optimizing}
\begin{array} {lll}
\min: & y & \\
\text{subject to:} & y\geq \frac{p_{i}C+\sum\limits_{k=i+1}^{n-1}p_{k}l_{k}}{\sum\limits_{k=i+1}^{n}p_{k}l_{k}}&i=2,3,\ldots,n-1 \\
& p_{\min} \leq p_{i} \leq p_{\max} & i=2,3,\ldots,n-1\\
\mathrm{var}: & y,~p_{i}&i=2,3,\ldots,n-1  \\
\end{array}
\end{equation*}

The hessian matrix of the non-equality constraints is:
\begin{equation*}
\left(\begin{array}{ccccccc}
0 & \cdots & 0 & 0 & \cdots & 0 & 0 \\
\vdots & \vdots & \vdots & \vdots & \vdots & \vdots & \vdots \\
0 & \cdots & 0 & 0 & \cdots & 0 & 0 \\
0 & \cdots & 0 & 0 & \cdots & 0 & -l_{i+1} \\
\vdots & \vdots & \vdots & \vdots & \vdots & \vdots & \vdots \\
0 & \cdots & 0 & 0 & \cdots & 0 & -l_{n-1} \\
0 & \cdots & 0 & -l_{i+1} & \cdots & -l_{n-1} & 0 \\
\end{array}\right)
\end{equation*}

Given $l_{k} > 0$ for $k=1,2,\cdots,n-1$, the hessian matrix is positive semi-definite. Then we can verify that the nonlinear constraints are convex.

By partially dualizing on the first set of non-equality constrains, we can get the following Lagrangian function:
\begin{equation*}
L(y,\bm{\mu})=y+\sum_{i=2}^{n-1}\mu_{i}(p_{i}C+\sum\limits_{k=i+1}^{n-1}p_{k}l_{k}-y\sum\limits_{k=i}^{n}p_{k}l_{k})
\end{equation*}
where $\mu_{i},~i=2,3,\ldots,n-1$ are the dual variables associated with the non-equality constraints.

Then the first-order optimality necessary condition is:
\begin{equation*}
\begin{split}
&1-\sum_{i=2}^{n-1}\mu_{i}\sum\limits_{k=2}^{n}p_{k}l_{k}=0 \\
&\mu_{i+1}C-l_{i}\sum\limits_{k=2}^{i}(y-1)\mu_{k} = 0,~i=2,3,\ldots,n-1 \\
\end{split}
\end{equation*}

Note that $y>1$, $0< l_{i}<C$ and $p_{\text{min}}\leq p_{i}\leq p_{\text{max}}$ for $i=2,3,\ldots,n-1$. Then by above equation, we can derive that $\mu_{i}>0,~i=2,3,\cdot,n-1$. Moreover, according to the complementary slackness condition:
\begin{equation*}
\mu_{i}(p_{i-1}C+\sum\limits_{k=i}^{n-1}p_{k}l_{k}-y\sum\limits_{k=i}^{n}p_{k}l_{k}) =0,~i=2,3,\ldots,n-1
\end{equation*}
we can get that the followings always hold:
\begin{equation*}
p_{i-1}C+\sum\limits_{k=i}^{n-1}p_{k}l_{k}-y\sum\limits_{k=i}^{n}p_{k}l_{k}=0,~i=2,3,\ldots,n-1
\end{equation*}
or
\begin{equation*}
\frac{p_{i-1}C+\sum\limits_{k=i}^{n-1}p_{k}l_{k}}{\sum\limits_{k=i}^{n}p_{k}l_{k}}=y,~i=2,3,\ldots,n-1
\end{equation*}
Then we complete the proof.

\subsection{proof of Lemma \ref{lem:pn}}
\begin{proof}
If $p_{n-1}\nrightarrow p_{n}$ for the optimal $g(z)$, we can assume there's another function $g'(z)$ which contains $n+1$ steps. We let $g'(z)=g(z)$ for $z=1,2,\ldots,b_{n-1},b_{n-1}+1,b_{n-1}+2,\ldots,b_{n}$ and $g'(b_{n-1}+1)=(p_{n-1}+p_{n})/2$. According to Lemma \ref{lem:lcr}, easily we can get that the competitive ratio with $g'(z)$ under the subsets satisfy the followings:
\begin{equation*}
\begin{split}
&\frac{\frac{p_{n-1}+p_{n}}{2}C}{p_1(l_n-1)}<\frac{p_{2}C}{p_1l_n} \\
&\frac{p_{i}C+\sum\nolimits_{k=i+1}^{n-1}p_kl_k+\frac{p_{n-1}+p_{n}}{2}}{\sum\nolimits_{k=i+1}^{n-1}p_kl_k+\frac{p_{n-1}+p_{n}}{2}+p_{1}(l_{n}-1)}<\frac{p_{i}C+\sum\nolimits_{k=i+1}^{n-1}p_kl_k}{\sum\nolimits_{k=i+1}^{n}p_kl_k}
\end{split}
\end{equation*}

That means the competitive ratio with $g'(z)$ is smaller than that with $g(z)$, contradicting the assumption that $g(z)$ is the optimal. This completes the proof.
\end{proof}

\subsection{proof of Theorem \ref{thm:2}}
\begin{proof}
Following similar lines with Lemma \ref{lem:nes_con} and \ref{lem:mp}, we can find that the worst instance for \ocsmb is similar to \socs.
For any worst instance $\omega$, if $z^{\omega}_{\ocsmb}(t)>z^{\omega}_{\socs}(t)(1+1/m)$, it can be verified that $x_{\socs}(t)<x_{\ocsmb}(t),~t\in\mathcal{T}$ must hold according to the algorithm rules.
So the cumulative profit difference between $\ocsmb$ and $\socs$ at time slot $t$ is bounded by $p(t)z^{\omega}_{\socs}(t)/m$ under $\omega$. So we have:
\begin{equation*}
\begin{split}
\cratio(\ocsmb)&\leq \max\limits_{t\in \mathcal{T}}\frac{R_{\ofa}(\omega)+p(t)z^{\omega}_{\socs}(t)/m}{R_{\socs}(\omega)} \\
&\leq \cratio(\socs) + \max\limits_{t\in \mathcal{T}} \frac{p(t)z^{\omega}_{\socs}(t)/m}{R_{\socs}(\omega)}
\end{split}
\end{equation*}

Moreover, for any time slot $t$, the following inequality always hold:
\begin{equation*}
\begin{split}
\frac{p(t)z^{\omega}_{\socs}(t)/m}{R_{\socs}(\omega)}&\leq \frac{g((1-1/m)z^{\omega}_{\socs}(t))z^{\omega}_{\socs}(t)/m}{\int_{z^{\omega}_{\socs}(t)}^{C}g(z)dz} \\
&\leq \frac{g((1-1/m)c^{\mathsf{th}})(C-c^{\mathsf{th}})/m}{\int_{c^{\mathsf{th}}}^{C}g(z)dz} \\
&\leq \frac{\theta \cratio(\socs)}{m^2}
\end{split}
\end{equation*}
The above inequality use the result that $p(t)\leq g((1-1/m)z^{\omega}_{\socs}(t))$ under the worst case $\omega$ and that $g(z)$ is convex.
\end{proof}

\end{document}